\documentclass[journal]{IEEEtran}

\usepackage{xcolor,soul,framed} %,caption

\colorlet{shadecolor}{yellow}
\usepackage[pdftex]{graphicx}
\graphicspath{{../pdf/}{../jpeg/}}
\DeclareGraphicsExtensions{.pdf,.jpeg,.png}

\usepackage[cmex10]{amsmath}
\usepackage{array}
\usepackage{mdwmath}

\usepackage{mdwtab}
\usepackage{eqparbox}
\usepackage{url}
\usepackage{amsfonts,amssymb}
\usepackage[ruled,linesnumbered]{algorithm2e}
\usepackage{algpseudocode}
\usepackage{color}
\usepackage{booktabs}
\usepackage[caption=false,subrefformat=parens]{subfig}
\usepackage{threeparttable}
\usepackage{graphicx}
\usepackage{mathtools}
\usepackage{amsmath}
\usepackage{siunitx}
\newtheorem{definition}{Definition} 
\newtheorem{theorem}{Theorem}
\newtheorem{assumption}{Assumption}

\begin{document}
\bstctlcite{IEEEexample:BSTcontrol}
\title{Controllability Test for Nonlinear Datatic Systems}
\author{
    \IEEEauthorblockN{Yujie Yang, Letian Tao, Likun Wang, Shengbo Eben Li
    \thanks{This study is supported by National Key R\&D Program of China with 2022YFB2502901 and NSF China under 52221005. Y. Yang and L. Tao contributed equally to this work. All correspondence should be sent to S. E. Li with email: lishbo@tsinghua.edu.cn.}
    \thanks{Y. Yang, L. Tao, L. Wang, and S. E. Li are with the School of Vehicle and Mobility, Tsinghua University, Beijing, 100084, China. Email: \{yangyj21, tlt22, wlk23\}@mails.tsinghua.edu.cn, lishbo@tsinghua.edu.cn.}}
}

% ====================================================================
\maketitle

% \markboth{IEEE TRANSACTIONS ON MICROWAVE THEORY AND TECHNIQUES, VOL.~60, NO.~12, DECEMBER~2012
% }{Roberg \MakeLowercase{\textit{et al.}}: High-Efficiency Diode and Transistor Rectifiers}

% === ABSTRACT ====================================================================
% =================================================================================
\begin{abstract}
Controllability is a fundamental property of control systems, serving as the prerequisite for controller design. While controllability test is well established in modelic (i.e., model-driven) control systems, extending it to datatic (i.e., data-driven) control systems is still a challenging task due to the absence of system models. In this study, we propose a general controllability test method for nonlinear systems with datatic description, where the system behaviors are merely described by data. In this situation, the state transition information of a dynamic system is available only at a limited number of data points, leaving the behaviors beyond these points unknown. Different from traditional exact controllability, we introduce a new concept called $\epsilon$-controllability, which extends the definition from point-to-point form to point-to-region form. Accordingly, our focus shifts to checking whether the system state can be steered to a closed state ball centered on the target state, rather than exactly at that target state. Given a known state transition sample, the Lipschitz continuity assumption restricts the one-step transition of all the points in a state ball to a small neighborhood of the subsequent state. This property is referred to as one-step controllability backpropagation, i.e., if the states within this neighborhood are $\epsilon$-controllable, those within the state ball are also $\epsilon$-controllable. On its basis, we propose a tree search algorithm called maximum expansion of controllable subset (MECS) to identify controllable states in the dataset. Starting with a specific target state, our algorithm can iteratively propagate controllability from a known state ball to a new one. This iterative process gradually enlarges the $\epsilon$-controllable subset by incorporating new controllable balls until all $\epsilon$-controllable states are searched. Besides, a simplified version of MECS is proposed by solving a special shortest path problem, called Floyd expansion with radius fixed (FERF). FERF maintains a fixed radius of all controllable balls based on a mutual controllability assumption of neighboring states. The effectiveness of our method is validated in three datatic control systems whose dynamic behaviors are described by sampled data.
\end{abstract}

% === KEYWORDS ====================================================================
% =================================================================================
% \begin{IEEEkeywords}
% Controllability, Datatic Systems
% \end{IEEEkeywords}

% \fi
% \IEEEpeerreviewmaketitle

% === I. INTRODUCTION =============================================================
% =================================================================================
\section{Introduction}
Feedback control plays a critical role in modern industry sectors, such as power electronic systems, chemical processes, and road transportation.
System analysis and controller synthesis are two principal tasks in feedback control systems.
The former involves studying inherent properties in plant dynamics, such as controllability, observability, and stability, while the latter is about designing an online controller to ensure that the closed-loop system exhibits desired behaviors.
As achieving desired behaviors depends on specific plant structures, excellent controller design requires an in-depth comprehension of some inherent properties of system dynamics.
Controllability, as a fundamental property of control systems, describes the system's ability to be steered from an initial state to an arbitrary final state in a finite time.
Before controller design, a proper controllability test should be performed to check whether the system has the required state transfer capability.

\begin{figure*}[!htbp]
  \begin{center}
  \includegraphics[width=0.7\linewidth]{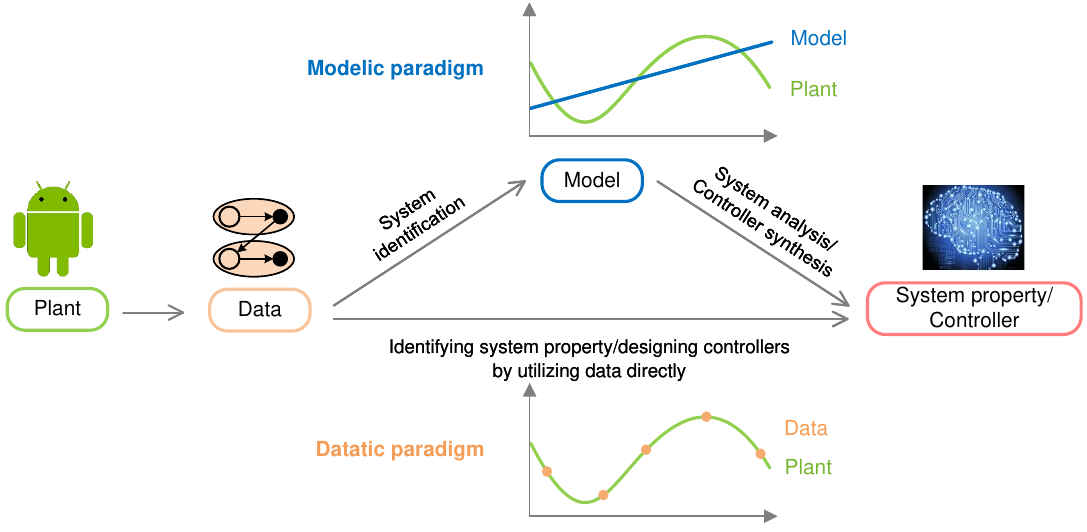}\\
  \caption{Comparison of control paradigms. In the modelic control paradigm, the first step is to establish a dynamic model through system identification. This model offers a continuous but inaccurate description of state transition information. In the datatic control paradigm, data is used directly for system analysis and controller synthesis, providing a discrete yet precise description of state transition.}
  \label{fig: paradigm}
  \end{center}
\end{figure*}

The concept of controllability was first introduced by Kalman and other contemporary researchers in the 1960s, along with a series of criteria for assessing the controllability of linear systems \cite{kalman1960on}.
The controllability test between two states can be reframed as verifying the existence of a solution for all admissible control inputs in a multi-step state transfer equation.
Linear systems are featured with the superposition principle, which means that an arbitrary subsequent state can be expressed as a linear combination of the initial state and control inputs.
Thus, its controllability test can be transformed into the full-rank test of a coefficient matrix that comes from linear state space model \cite{trentelman2012control}.
This coefficient matrix is also known as the controllability matrix.
In addition, some researchers have proposed another coefficient matrix called controllability Gramian, whose positive definiteness is used as an equivalent criterion for full-rank test.
In contrast to these two algebraic criteria, which analyze controllability in an integrated manner, the modal criteria focus on finding controllable subsystems or controllable modes in a regional way.
For example, Kalman decomposition (1963) linearly transformed a state space model into a standard form, which could be decomposed into observable and controllable subsystems~\cite{kalman1963mathematical}.
Popov et al. (1966) introduced the famous Popov-Belevitch-Hautus test, which identified uncontrollable modes of a linear system by finding left eigenvectors of the system matrix that is orthogonal to the input matrix~\cite{popov1966hiperstabilitatea}.
In short, linear controllability test is equivalent to property assessment of certain matrices constructed from state space models.
When it comes to nonlinear systems, controllability test requires analyzing the existence of solutions to nonlinear ordinary differential equations, which depends on specific forms of the equations and is undeterminable in general cases.
Due to this difficulty, only a few sufficient conditions have been proposed for nonlinear controllability test.
Moreover, these sufficient conditions are given on a case-by-case basis due to the complexity and diversity of system dynamic behaviors.
Gershwin et al. (1971) provided a sufficient condition for global controllability of nonlinear systems by utilizing a Lyapunov-like scalar function to indicate whether the system state can be transferred to the target state \cite{gershwin1971controllability}.
Hermann (1977) used the Lie theory to characterize the directions of vector fields in continuous-time nonlinear systems. 
If the vector fields span the entire state space, system controllability is ensured under some special cases, such as symmetric and affine systems~\cite{hermann1977nonlinear}.
Yamamoto (1977) transformed controllability test to solving a fixed-point problem of nonlinear continuous operator in a Banach space.
With this transformation, his study discovered that a sufficient condition for controllability is the existence of a subset in the Banach space that is invariant for the nonlinear operator~\cite{yamamoto1977controllability}.
The above analysis reveals a key requirement shared by both linear and nonlinear systems: an accurate mathematical model of the system is indispensable.
For example, Gershwin's method requires an affine nonlinear model to specify gradient directions for validating Lyapunov conditions.
Hermann's method utilizes an infinitely differentiable model to compute vector fields of closed-loop systems.
Yamamoto’s method relies on a mathematical model consisting of norm-bounded functions to discern the structure of dynamic equation and transform it into a set of linear equations.

Recent years have witnessed a paradigm shift in control field from \textit{modelic} description to \textit{datatic} description.
Here, \textit{modelic} and \textit{datatic} are two newly coined words, where \textit{modelic} means model-driven, model-based or model-related, and \textit{datatic} has a corresponding meaning related to data.
The key distinction between these two paradigms lies in how to describe system dynamics: either a mathematical model or data points \cite{yang2024stability, zhan2024canonical}.
In the modelic control paradigm, an explicit system model is established based on physical laws or system identification, providing a continuous description of system dynamics across the state-action space.
Such models are essential for analyzing system properties and synthesizing controllers, as exemplified by the aforementioned controllability test methods.
However, in many complex systems such as earth atmosphere, financial market, and road transportation, it is often impossible to accurately capture their dynamics solely through basic physical laws or simple hypothetical functions.
As a result, strong simplification of their dynamics must be made at the cost of introducing substantial model errors, which can significantly sacrifice the accuracy of system property analysis.
The datatic control paradigm, fueled by technical advancements of data storage and parallel computation, has shown great potential in solving complex control tasks.
This control paradigm directly leverages data for system property analysis and controller synthesis, eliminating the dependency on pre-built mathematical models.
As direct measurements of system's input and output sequences, data provides accurate state transition information, thus bypassing the issue of model mismatch.
Nevertheless, the accuracy in datatic description comes at the expense of tempo-spatial incompleteness.
Unlike explicit models, system measurements are not continuous in both temporal and spatial domains but only in the form of a limited number of data points.
There is no information in the interval of any two data points, as shown in Fig. \ref{fig: paradigm}.
Therefore, datatic description of a dynamic system must be discrete rather than continuous in the state-action space.
As a consequence, traditional system property analysis tools based on continuous models become inapplicable in datatic control systems.
Since controller design depends on whether states are controllable, datatic controllers become meaningless without reliable controllability test.

Despite its fundamental importance, controllability test is often neglected in systems with datatic description.
While numerous data-based learning methods emerge, e.g., imitation learning~\cite{hussein2017imitation, ho2016generative} and reinforcement learning~\cite{li2023reinforcement, guan2021direct}, most of them share a default but actually unconfirmed assumption that all the system states are controllable.
This neglect has changed in recent years, with some researchers beginning to consider the controllability test of linear datatic systems.
One class of datatic test methods verifies controllability through experiments that strictly limit control inputs.
Wang et al. (2011) proposed to conduct $m$ groups of experiments, where control inputs were set to the same one-hot vector in each test, with $m$ indicating the dimension of control input \cite{zhuowang2011databased}.
The system controllability was assessed by examining the rank of a new controllability matrix generated from these specialized experiments.
Subsequently, Liu et al. (2014) extended this experimental method by considering both measurement noise and process noise and generalized the one-hot control input to a linearly independent vector via invertible linear transformation \cite{liu2014databased}.
Shaker et al. (2017) performed similar specialized experiments to compute a multi-step controllability Gramian for controllability test \cite{shaker2017new}.
The inverse of the smallest eigenvalue of controllability Gramian serves as the upper bound on the minimum control energy.
The core idea of these three methods is to calculate controllability matrix from data collected from a finite number of experiments, in which control inputs are held constant in each experiment group.
Another class of datatic test methods for linear systems does not perform experiments with fixed control inputs but relies on trajectories generated by experiments with arbitrary policies.
Waarde et al. (2020) proposed the concept of data informativity to assess whether the collected data is informative enough for determining controllability within an unknown system set \cite{vanwaarde2020data}. 
They introduced a data-driven Hautus test, which built a trajectory-based matrix and examined whether this matrix was fully ranked or not. 
Mishra et al. (2021) extended this arbitrary-policy-based experiment method to general input-output linear systems that assume no measurement of system state. 
Under the condition of persistently exciting control inputs, they provided a sufficient and necessary condition for system controllability by examining the rank of Hankel matrix constructed from the input-output trajectories \cite{mishra2021datadriven}.
The aforementioned studies reveal that controllability test in linear systems can be simplified to validating properties of certain matrices from experiments, which is feasible to implement in datatic control.

Obviously, existing datatic test methods are limited to linear systems.
This limitation is caused by the fact that traditional controllability is actually a kind of exact controllability (also called point-to-point controllability), which requires continuous state transition information to accurately position initial and target states, as well as their intermediate states.
In a nonlinear system with datatic description, its information is always insufficient no matter how many data samples are collected.
Specifically, to match the needs of traditional controllability, one must find a sequence of data points that are exactly connected one after another from the initial state to the target state.
This is almost impossible because data points are discrete and may not be collected in the manner of complete trajectories.
As a result, exact controllability is not a practical definition in nonlinear datatic systems.
Moreover, strictly controlling a state to a given point is often unnecessary in many real-world tasks.
For example, in anesthesia control, doctors are concerned with whether drug can be injected into a target organ instead of exactly at a specific cell.
In autonomous driving, it is sufficient for the vehicle to stay within a certain lane, not necessarily aligned perfectly with the lane center.

To the best of our knowledge, this paper proposes the first method to test the controllability of nonlinear systems with datatic description, where state transition information is only available at a finite number of data points.
We introduce a new definition called $\epsilon$-controllability, which concerns whether the system state can be transferred into a small neighborhood of the target state rather than exactly at that state.
The Lipschitz continuity assumption is used to restrict the one-step transfer range of an unknown state to the neighborhood of the subsequent state in the dataset.
On its basis, one-step controllability backpropagation theorem is introduced, which propagates the $\epsilon$-controllability from a known state ball to a new one.
Leveraging this theorem, our proposed algorithm can maximumly find $\epsilon$-controllable states of a specified target state. 
Our main contributions are summarized as follows.
\begin{enumerate}
    \item We propose the concept of $\epsilon$-controllability that extends the definition of controllability from point-to-point form to point-to-region form. A state is considered to be $\epsilon$-controllable with respect to a given target state if it can be steered to a closed state ball of this target through finite-step state transition. On its basis, we introduce the notion of $\epsilon$-controllable set, in which all the states are $\epsilon$-controllable with respect to the same target state. For a nonlinear system with datatic description, its controllability test is equivalent to finding those states that belong to the $\epsilon$-controllable set of a particular state as many as possible.
    \item We propose a tree search algorithm called maximum expansion of controllable subset (MECS) to effectively identify potential $\epsilon$-controllable states in the dataset. The Lipschitz continuity is assumed to hold, which confines the one-step transition of all the states in a state ball to a known neighborhood of the subsequent state. This means that if all the states within this neighborhood are $\epsilon$-controllable, it follows that those within the state ball are also $\epsilon$-controllable. Therefore, $\epsilon$-controllability can propagate from one subset to another, which is referred to as one-step controllability backpropagation. Starting from a target $\epsilon$-controllable subset, our MECS algorithm iteratively applies one-step controllability backpropagation to gradually identify all the states in the dataset that are capable of reaching this target subset.
\end{enumerate}

% === II. Definition of $\epsilon$-Controllability ================================
% =================================================================================
\section{Definition of $\epsilon$-Controllability}
The traditional definition of system controllability is the ability to steer the system from a certain initial state to the target state within a finite time step, which is referred to as point-to-point controllability.
This concept, which concerns achieving precise state transition, is also called \textit{exact controllability}.
Existing research on exact controllability tests relies on continuous state transition information across the entire state-action space, which can be provided by a known system model.
However, in a datatic system, data points provide a discrete representation of state transition, leaving the dynamics information within the intervals between any two points unknown.
This information discontinuity makes the verification of exact controllability impractical for systems with datatic description.
Consequently, for their controllability, our interest lies in whether the system state can be transferred to a small region of the target point.

% =======
% FIG. 02
% =======
\begin{figure}[!htbp]
  \begin{center}
  \includegraphics[width=3.0in]{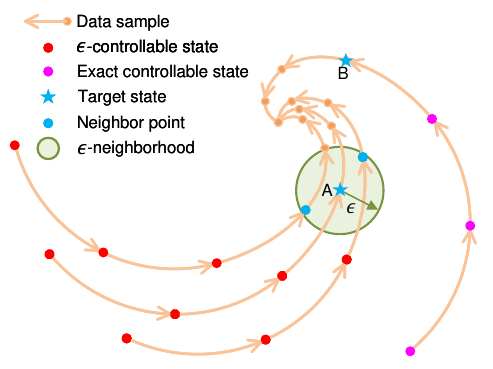}\\
  \caption{The difference between exact controllability and $\epsilon$-controllability: The graph depicts four state trajectories from collected data. Stars denote target states; blue points denote states within the neighborhood of the target $A$ with radius $\epsilon$. Pink points denote exact controllable states w.r.t. target $B$, as they can be precisely steered to the target state. Conversely, red points denote $\epsilon$-controllable states w.r.t. target $A$, which are only capable of being steered into a neighborhood of that target state.}
  \label{fig: data sample}
  \end{center}
\end{figure}

Following the analysis above, we introduce a new definition of controllability, called \textit{$\epsilon$-controllability}, which is actually a kind of point-to-region controllability.
Under this definition, the system controllability concerns whether the state can be transferred to a small neighborhood of the target state rather than exactly at that state point.
The difference between exact controllability and $\epsilon$-controllability is shown in Fig. \ref{fig: data sample}.

Consider a discrete-time nonlinear system:
\begin{equation}
\label{equ: dynamics}
    x^{\prime}=f(x,u),
\end{equation}
where $x\in\mathcal{X}$ is the system state, $x^{\prime}\in\mathcal{X}$ is the next state, $u\in\mathcal{U}$ is the control input, and $\mathcal{X}$ and $\mathcal{U}$ are the state and control input spaces, respectively.
Before formally defining $\epsilon$-controllability, let us define a neighborhood of a state, which is a closed ball in the state space:
\begin{equation}
    \mathcal{B}(y,\delta)=\{z\in\mathcal{X}\mid d(z,y)\leq\delta\},
\end{equation}
where $y\in\mathcal{X}$ is the ball center, $\delta \geq 0$ is the ball radius, and $d(\cdot,\cdot)$ is a metric on the state space $\mathcal{X}$, which is chosen as the Euclidean distance in this paper.
We refer to $\mathcal{B}(y,\delta)$ as a state ball of $y$ with radius $\delta$.

\begin{definition}[$\epsilon$-controllability]
A state $x\in\mathcal{X}$ is said to be $\epsilon$-controllable with respect to its target $x_{\text{T}}$ if $x$ can reach the target neighborhood $\mathcal{B}(x_{\text{T}},\epsilon)$ within a finite number of steps, where $\epsilon$ denotes the error radius.
\end{definition}

\begin{definition}[$\epsilon$-controllable set]
The $\epsilon$-controllable set refers to a set of states that are $\epsilon$-controllable with respect to the same target $x_{\text{T}}$.
\end{definition}

To provide a mathematical description of $\epsilon$-controllability, let us introduce a reachability function $g(\cdot,\cdot):\mathcal{X}\times\mathcal{X}\to\{0,1\}$:
\begin{equation}
    g(x,x_{\text{T}})=\max_{0\leq t<\infty} \{\mathbb{I}(x_t=x_{\text{T}})|x_0=x\},
\end{equation}
where $\{x_i\}_{i=0}^{\infty}$ denotes the state trajectory, and $\mathbb{I}(\mathrm{\cdot})$ is an indicator function. The reachability function $g(x,x_{\text{T}})=1$ if there exists a control policy driving $x$ to $x_{\text{T}}$ within a finite number of steps. Conversely, $g(x,x_{\text{T}})=0$ indicates that no policy can drive $x$ to $x_{\text{T}}$. 
In mathematics, $x$ is $\epsilon$-controllable with respect to $x_{\text{T}}$ if and only if 
\begin{equation}
\exists z\in\mathcal{B}(x_{\text{T}},\epsilon), g(x,z)=1.
\end{equation}
The $\epsilon$-controllable set $\mathcal{C}(x_{\text{T}},\epsilon)$ is denoted as
\begin{equation}
    \mathcal{C}(x_{\text{T}},\epsilon)=\{x\in\mathcal{X}|\exists z\in\mathcal{B}(x_{\text{T}},\epsilon), g(x,z)=1\}.
\end{equation}

The $\epsilon$-controllability is an extension of exact controllability from a point-to-point form to a point-to-region form.
When the value of error radius equals zero, $\epsilon$-controllability degenerates to exact controllability.
Obviously, $\epsilon$-controllability proves particularly suitable for datatic systems with discrete and sparse state transition information.
This concept also provides a tolerable error bound for reaching the target state, which is determined by the error radius.
For many practical control tasks, it is adequate to reach a state that is close to the target but with some small error, making $\epsilon$-controllability a reasonable definition.

% === Ⅲ. Controllability Test for Nonlinear Datatic Systems ======================
% =================================================================================
\section{Controllability Test for Nonlinear Datatic Systems}
A datatic system is represented by a dataset $\mathcal{D}$ collected from system dynamics \eqref{equ: dynamics}:
\begin{equation}
\mathcal{D}=\{(x_i,u_i,x_i^{\prime})|x_i^{\prime}=f(x_i,u_i),i=1,2,\cdots, N\},
\end{equation}
where $N$ is the number of data points.
The $\epsilon$-controllability test aims to find all $\epsilon$-controllable states in $\mathcal{D}$
with respect to a specific target state.
Recall that exact controllability test is impractical in a datatic system because an arbitrarily chosen target may not be located in the dataset $\mathcal{D}$.
The $\epsilon$-controllability allows us to relax the requirement of exactly reaching the target state.
Even with this relaxation, $\epsilon$-controllability test is still intractable without further assumption on system dynamics.
This is because those intermediate points in the sequence still have to be exactly connected, which is not always accessible due to the discreteness of datatic description.
What we need is to further relax exact connection to some neighboring connection.
That is to say, any state in the sequence does not have to reach its subsequent state exactly but can just reach a neighborhood of it.
To test $\epsilon$-controllability using such a data sequence, we need to verify whether all states in all neighborhoods will reach their subsequent neighborhoods in one step.
This verification requires extending state transition information from data points to their neighborhoods, which necessitates a continuity assumption of system dynamics.
With a proper continuity assumption, we are able to iteratively backpropagate $\epsilon$-controllable neighborhoods to their one-step predecessors, which builds the basis of identifying all $\epsilon$-controllable states in the dataset.

\subsection{Lipschitz continuity assumption}
As mentioned above, testing $\epsilon$-controllability in general datatic systems requires continuity assumption of system dynamics $f$.
Specifically, we need to ensure that two states in the same neighborhood are still close to each other after a one-step transition, so that both of them can be contained in a subsequent neighborhood.
This means that we need to restrict the rate of change of $f$, which can be achieved by forcing $f$ to be continuously differentiable.
Here, we assume that $f$ is Lipschitz continuous, which is weaker than continuous differentiability but achieves the same goal.
\begin{assumption}[Lipschitz continuity]
\label{asp: Lipschitz}
The system dynamics $f(x,u)$ is Lipschitz continuous with respect to $x$ and $u$, i.e., there exists Lipschitz constants $L_x,L_u\geq0$ such that for all $x_1, x_2 \in \mathcal{X}$ and $u_1, u_2\in\mathcal{U}$, the following inequality holds:
\begin{equation}
d(f(x_1,u_1), f(x_2,u_2))
\leq L_x d(x_1, x_2) +  L_u d(u_1, u_2),
\end{equation}
where $d(\cdot,\cdot)$ is the metric on $\mathcal{X}$ and $\mathcal{U}$. In this paper, it is chosen as the Euclidean distance.
\end{assumption}

Assumption \ref{asp: Lipschitz} enables us to extend state transition information from discrete data points to their continuous neighborhoods.
With this extension, we can perform controllability test based on a sequence of connected neighborhoods instead of just data points, thus overcoming the difficulty of data discreteness.

% method and algorithm
\subsection{One-step controllability backpropagation}
The test of controllability relies on one of its basic properties: transitivity.
This property means that if $x_1$ is controllable with respect to $x_2$, and $x_2$ is controllable with respect to $x_3$, then $x_1$ is also controllable with respect to $x_3$.
When testing controllability, transitivity is actually applied in the opposite direction of a trajectory.
Specifically, the first controllable state we know is the target state.
Then, all states that reach the target state in one step are termed as controllable.
By repeating this process, all controllable states in $\mathcal{D}$ can be found.
This process can be viewed as propagating controllability backward along state trajectories and is therefore called \textit{controllability backpropagation}.
Controllability test in a datatic system needs to extend the exact state transition to a neighborhood transition.
In other words, this test should be performed in a backward manner not only along the data points themselves but also along all states in their neighborhoods.

To understand why such backpropagation among neighborhoods is possible, let us first consider a one-step controllability backpropagation.
Suppose that all states in the neighborhood of a data point $z$ with radius $\sigma$ are $\epsilon$-controllable, i.e., $\mathcal{B}(z,\sigma)\subset\mathcal{C}(x_\text{T},\epsilon)$.
According to the transitivity of controllability, if a state reaches this neighborhood in one step, it is also controllable.
Let us assume that $x_i$ is such a state in the dataset.
To continue backpropagation, we need to further find a neighborhood of $x_i$, denoted as $\mathcal{B}(x_i,r_i)$, in which all states are controllable.
Here, $r_i$ is the radius of new state ball.
This is possible if all states in $\mathcal{B}(x_i,r_i)$ reach $\mathcal{B}(z,\sigma)$ in one step.
Thanks to the Lipschitz continuity assumption, we can always find an $r_i\ge0$ that achieves this goal.
This conclusion is formally stated in the following one-step controllability backpropagation theorem.

\begin{theorem}[one-step controllability backpropagation]
\label{thm: backprop}
Given a neighborhood $\mathcal{B}(z,\sigma)$ in which all states are $\epsilon$-controllable with respect to $x_{\text{T}}$, if there exists a data point $x_i$ such that its subsequent state $x_i^\prime$ lies within $\mathcal{B}(z,\sigma)$, i.e., $x_i^\prime\in\mathcal{B}(z,\sigma)$, then all states in $\mathcal{B}(x_i,r_i)$ are $\epsilon$-controllable with respect to $x_{\text{T}}$, where
$$r_i=(\sigma-d(x_i^\prime,z))/L_x.$$
\end{theorem}

\begin{proof}
Considering that $x_i^\prime\in\mathcal{B}(z,\sigma)$, we have
\begin{equation}
d(x^\prime_i, z)\leq\sigma.
\end{equation}
For any state $x\in\mathcal{B}(x_i,r_i)$, we have
\begin{equation}
    d(x,x_i)\leq(\sigma-d(x_i^\prime,z))/L_x.
\end{equation}
Given the Lipschitz continuity of $f$ with respect to $x$, we have the following inequality:
\begin{equation}
\begin{aligned}
    d(f(x, u_i), x_i^\prime)
    &=d(f(x, u_i), f(x_i, u_i))\\
    &\leq L_x d(x, x_i)\\
    &\leq \sigma-d(x_i^\prime, z).
\end{aligned}
\end{equation}
Then using the triangle inequality, we obtain that
\begin{equation}
    d(f(x,u_i), z)\leq d(f(x, u_i), x_i^\prime) + d(x_i^\prime, z)\leq\sigma.
\end{equation}
Therefore, for any state $x\in\mathcal{B}(x_i,r_i)$, we have $f(x, u_i)\in\mathcal{B}(z,\sigma)$. This implies that $x$ reaches an $\epsilon$-controllable subset with respect to $x_{\text{T}}$ in one step, confirming that $x$ is also $\epsilon$-controllable.
\end{proof}

\begin{figure}[!htbp]
  \begin{center}
  \includegraphics[width=80mm]{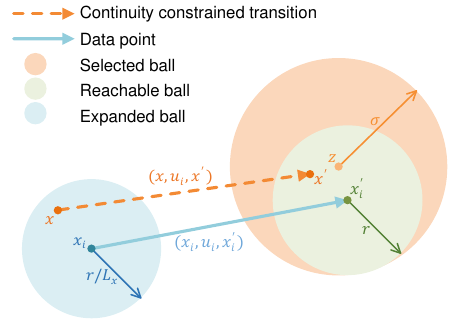}\\
  \caption{One-step controllability backpropagation. Lipschitz continuity restricts the subsequent states of the expanded ball to be within the reachable ball. If the reachable ball is contained by an $\epsilon$-controllable subset, i.e., the selected ball, all states in the expanded ball are also $\epsilon$-controllable. Consequently, controllability backpropagates from the selected ball to the expanded ball.}
  \label{fig: controllability one-step backpropagation}
  \end{center}
\end{figure}

Theorem \ref{thm: backprop} is the basis of $\epsilon$-controllability test in general datatic systems.
It extends the transitivity of controllability from exact state transition to neighborhood transition using the Lipschitz continuity of system dynamics.
Fig. \ref{fig: controllability one-step backpropagation} gives an illustration of Theorem \ref{thm: backprop}.
The orange ball represents the known neighborhood of $\epsilon$-controllable states $\mathcal{B}(z,\sigma)$.
It is called \textit{selected ball} because there may be many such neighborhoods and this is the one we select for backpropagation.
The blue point represents the state $x_i$ that reaches $\mathcal{B}(z,\sigma)$ in one step, and its subsequent state $x_i^\prime$ is represented by the green point.
According to Lipschitz continuity, all states in a neighborhood of $x_i$, called \textit{expanded ball}, will reach a neighborhood of $x_i^\prime$ in one step, called \textit{reachable ball}.
The name of expanded ball comes from the fact that it is an expansion of the current $\epsilon$-controllable subset.
The name of reachable ball comes from the fact that it is where the states in the expanded ball will reach in one step.
The radii of the two balls are related by the Lipschitz constant $L_x$.
The maximum radius of the reachable ball is achieved when it is tangent to the selected ball, i.e., $r=\sigma-d(x_i^\prime, z)$.
At this time, the radius of the expanded ball also achieves its maximum $r/L_x$, which exactly equals $r_i$ in Theorem \ref{thm: backprop}.

The one-step controllability backpropagation theorem provides a solid basis for identifying new $\epsilon$-controllable subsets.
By iteratively applying this theorem, we can find all $\epsilon$-controllable subsets in the dataset, which is basically what our following $\epsilon$-controllability test algorithm does.

\subsection{Controllability test algorithm for datatic system}
This section introduces an $\epsilon$-controllability test algorithm for general datatic systems, called maximum expansion of controllable subset (MECS).
This algorithm finds all $\epsilon$-controllable states in the dataset by searching a tree with $\epsilon$-controllable balls as its nodes, which is called $\epsilon$-controllable tree.
The root node of the tree is the target ball, which is $\epsilon$-controllable by definition.
Every other node of the tree is a ball centered at a data point.
The child nodes are determined by the state transition relationship: all states in a child node reach its parent node in one step.
According to Theorem \ref{thm: backprop}, such child nodes are also $\epsilon$-controllable.
For each data point, it is $\epsilon$-controllable if it is contained in at least one node of the tree.
By traversing the $\epsilon$-controllable tree, we can obtain all $\epsilon$-controllable states in the dataset.
In the following description, we use the terms ``node" and ``ball" interchangeably, which actually refer to the same thing.

\begin{figure*}[!htbp]
  \begin{center}
  \includegraphics[width=0.9\linewidth]{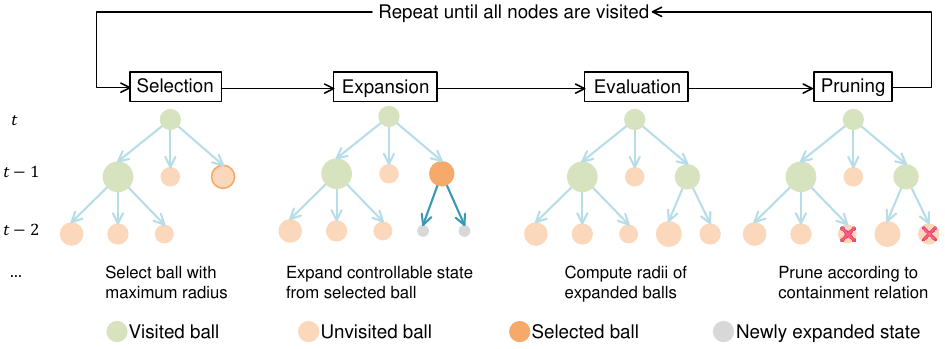}\\
  \caption{Four key steps of MECS algorithm. \textbf{Selection}: Choose the leaf node with the maximum radius. \textbf{Expansion}: Find data points with subsequent states within the selected ball. \textbf{Evaluation}: Compute the radii of the expanded balls using the one-step controllability backpropagation theorem.  \textbf{Pruning}: Remove the leaf nodes that are contained by others.}
  \label{fig: expansion tree}
  \end{center}
\end{figure*}

The MECS algorithm searches the $\epsilon$-controllable tree by iteratively performing four steps: a) Selection, b) Expansion, c) Evaluation, and d) Pruning, as shown in Fig. \ref{fig: expansion tree}.
In the selection step, we choose one node from all current leaf nodes for expansion.
The selection rule is simple: we always choose the leaf node with the maximum radius.
Although any tree search algorithm can be used for selection, e.g., depth-first search and breadth-first search, we find that selecting the node with the maximum radius significantly improves algorithm efficiency.
This heuristic is based on the intuition that a larger state ball is likely to contain more data points, resulting in the acquisition of more expanded balls.
This is also where the name of our MECS algorithm comes from, i.e., it selects the $\epsilon$-controllable subset with the maximum radius for expansion.
In the expansion step, we traverse the dataset to find all data points whose subsequent states lie within the selected ball.
According to the transitivity of controllability, these data points are $\epsilon$-controllable.
Furthermore, according to Theorem \ref{thm: backprop}, neighborhoods of these data points are also $\epsilon$-controllable.
The radii of their neighborhoods are computed in the evaluation step.
With radii computed, the newly expanded balls become leaf nodes of the $\epsilon$-controllable tree and are ready for next expansion.
Before starting the next iteration, we perform a pruning step to eliminate overlap among leaf nodes because it will lead to redundant expansions.
Specifically, if a leaf node is contained in another one, then all states expanded from the smaller node can also be expanded from the larger one.
Therefore, the smaller node can be removed to improve computational efficiency.
Since other leaf nodes have been pruned in previous iterations, we only need to consider the newly expanded nodes in the pruning step: if an expanded node is contained by any other leaf node, it is removed; conversely, if an expanded node contains any other leaf node, the latter is removed.
With the pruning step, the advantage of selecting the leaf node with the maximum radius becomes clearer: it will likely result in larger and more expanded balls, allowing more leaf nodes to be pruned and saving more computing time.
MECS repeats these four steps until all nodes in the $\epsilon$-controllable tree are visited.
At this time, there are no more states to expand, and all $\epsilon$-controllable states are found.

The pseudocode of MECS is detailed in Algorithm \ref{alg: MECS}.
At its beginning, we estimate the local Lipschitz constants within a neighborhood of radius $\delta$ for each state $x_i$ in the dataset $\mathcal{D}$.
The estimation method will be introduced in the following section.
Next, two sets are initialized: $S_{\mathrm{unv}}$ for unvisited $\epsilon$-controllable balls and $S_{\mathrm{vis}}$ for visited ones, which correspond to leaf nodes and internal nodes of the $\epsilon$-controllable tree.
The target ball $\mathcal{B}(x_{\text{T}},\epsilon)$ is treated as an initial $\epsilon$-controllable ball and added to $S_{\mathrm{unv}}$.
Then, the aforementioned four steps are iteratively performed to search the $\epsilon$-controllable tree.
It is worth noting that in the evaluation step, the radius of the expanded ball must be smaller than $\delta$, i.e.,
$$r=\min\{\delta, \sigma-d(x_i^\prime, z)\}.$$
This is because the estimated local Lipschitz constant is valid only within a ball of radius $\delta$.

\begin{algorithm}
\caption{Maximum Expansion of Controllable Subset (MECS)}
\label{alg: MECS}
\SetKwInOut{Input}{Input}\SetKwInOut{Output}{Output}
\SetKw{Break}{break}
\Input{The target state $x_{\text{T}}$, error radius $\epsilon$, dataset $\mathcal{D}$}
\Output{A set of $\epsilon$-controllable balls $S_{\mathrm{vis}}$}
\BlankLine
\emph{// Lipschitz constant estimation}\\
\For{$i=1,2,\dots,N$}{
$(\hat L_{x_i}, \hat L_{u_i}) \gets \text{LipschitzEstimation}(i)$\;
}
Initialize two empty sets $S_{\mathrm{vis}}$ and $S_{\mathrm{unv}}$\;
$S_{\mathrm{unv}}.\mathrm{add}(\mathcal{B}(x_{\text{T}}, \epsilon))$\;
\While{$S_{\mathrm{unv}}$ is not empty}{
    \emph{// Selection}\\
    $\mathcal{B}(z,\sigma)\leftarrow \arg\max_{\mathcal{B}(z,\sigma)\in S_{\mathrm{unv}}}\sigma$\;
    Move $\mathcal{B}(z,\sigma)$ from $S_{\mathrm{unv}}$ to $S_{\mathrm{vis}}$\;
    \For{$i=1,2,\dots,N$}{
        \emph{// Expansion}\\
        \If{$x_i^\prime \in \mathcal{B}(z,\sigma)$}{
            \emph{// Evaluation}\\
            $r_i=\min\{\delta, (\sigma-d(x_i^\prime, z))/\hat{L}_{x_i}\}$\;
            $S_{\mathrm{unv}}.\mathrm{add}(\mathcal{B}(x_i, r_i))$\;
        }
    }
    \emph{// Pruning}\\
    Prune $S_{\mathrm{unv}}$ according to containment relation\;
    \textbf{return} $S_{\mathrm{vis}}$\;
}
\end{algorithm}

\begin{figure}[!htbp]
  \begin{center}
  \includegraphics[width=1\linewidth]{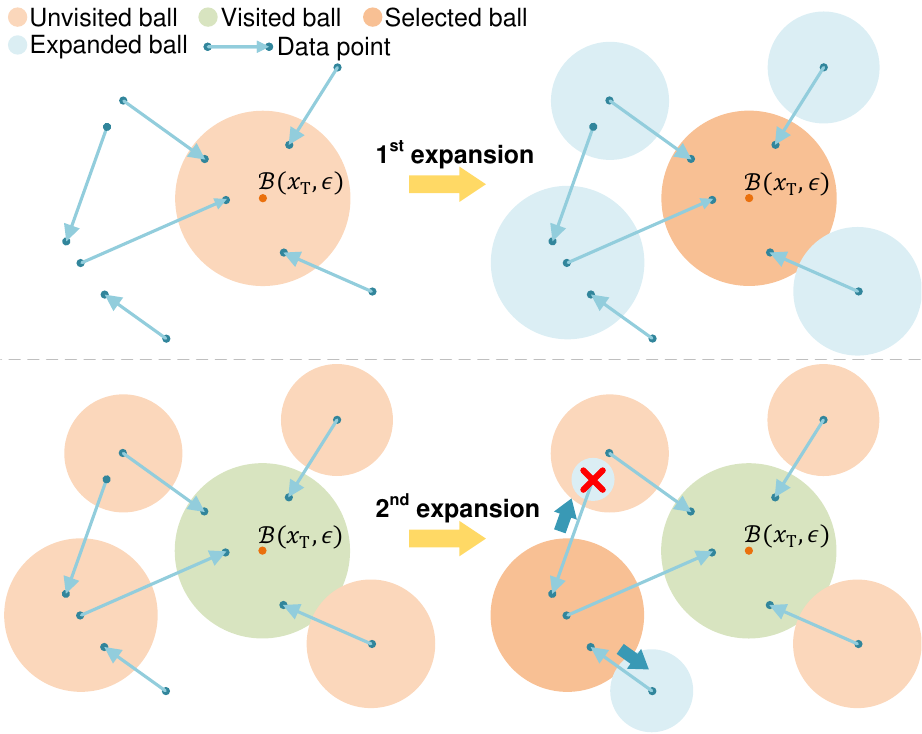}\\
  \caption{A example of MECS. In the first iteration, the target ball is selected, and four $\epsilon$-controllable balls are expanded. In the second iteration, two more $\epsilon$-controllable balls are expanded, and one of them is pruned.}
  \label{fig: controllability expansion algorithm}
  \end{center}
\end{figure}

Fig. \ref{fig: controllability expansion algorithm} demonstrates a simple example of how MECS works.
Initially, the target ball is the only unvisited $\epsilon$-controllable ball and is thus selected.
In the first expansion, four data points are found whose subsequent states lie within the target ball, and the radii of their corresponding $\epsilon$-controllable balls are computed.
Now the four newly expanded balls become unvisited balls.
In the second iteration, we select the one with the maximum radius among them and further obtain two expanded balls.
In the pruning step, one of them is removed because it is contained by an unvisited ball.

\subsection{Estimation of local Lipschitz constants}
In our MECS algorithm, computing the radius of the expanded ball relies on the local Lipschitz constant $L_x$ within a neighborhood of each data point.
For a data point $(x_i,u_i,x_i^\prime)$, its local Lipschitz constant is estimated by considering all data points within a neighborhood of $x_i$ with radius $\delta$, whose indexes are denoted by $$\mathcal{I}_i=\{j|d(x_i,x_j)\leq\delta,1\leq j\leq N\}.$$
The estimation of $L_x$ involves solving a linearly constrained quadratic programming (LCQP) problem:
\begin{equation}
\begin{aligned}
&\textrm{min}_{L_x,L_u}\quad L_x^2 + L_u^2 \\
&\textrm{s.t.}\quad (L_x, L_u)\in \mathcal{C}_\text{feas},
\end{aligned}
\label{equ: opt. for estimating Lips}
\end{equation}
where $\mathcal{C}_\text{feas}$ is the feasible region of LCQP:
\begin{equation}
\begin{split}
\mathcal{C}_\text{feas} = \{(L_x, L_u) \mid 
& d(x_j^\prime, x_k^\prime) \\
&\leq L_x d(x_j, x_k) + L_u d(u_j, u_k), \\
&\forall j,k \in \mathcal{I}_i\}.
\end{split}
\end{equation}
The estimated Lipschitz constants are the minimum values that satisfy the Lipschitz continuity constraints imposed by data points within the neighborhood.

\begin{figure}[!bp]
  \begin{center}
  \includegraphics[width=50mm]{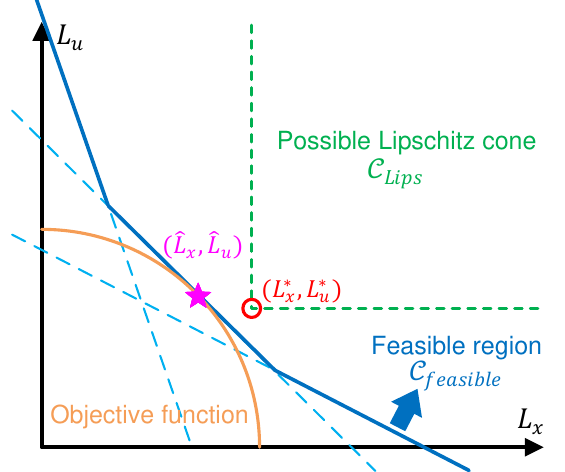}\\
  \caption{A geometric interpretation of the Lipschitz constant estimation problem. In the $L_x-L_u$ coordinate plane, the region enclosed by the green line represents all possible Lipschitz constants, and $L_{x_i}^*$ and $L_{u_i}^*$ denote the smallest possible local Lipschitz constants of $f(x,u)$ with respect to $x$ and $u$, respectively. The light blue line represents the linear constraints of LCQP, while the dark blue line encloses the feasible domain. The orange curve depicts the objective function of LCQP. In this illustration, $(\hat L_{x_i},\hat L_{u_i})$ is the optimal solution to LCQP.}
  \label{fig: opt. for estimating Lips}
  \end{center}
\end{figure}

Fig. \ref{fig: opt. for estimating Lips} shows a geometric interpretation of LCQP \eqref{equ: opt. for estimating Lips}. 
In the $L_x-L_u$ coordinate plane, all possible Lipschitz constants constitute the possible Lipschitz set:
\begin{equation}
\begin{split}
\mathcal{C}_\text{Lips} = \{(L_x, L_u) \mid 
& d(f(x_1,u_1), f(x_2,u_2)) \\
&\leq L_x d(x_1, x_2) + L_u d(u_1, u_2), \\
&\forall x_1, x_2 \in \mathcal{B}(x_i,\delta), u_1, u_2 \in \mathcal{U}\}.
\end{split}
\end{equation}
In fact, the possible Lipschitz set is a cone as follows:
\begin{equation}
\mathcal{C}_\text{Lips} = \{(L_x,L_u)|L_x\geq L_{x_i}^*,L_u\geq L_{u_i}^*\},
\end{equation}
where $L_{x_i}^*$ and $L_{u_i}^*$ are the smallest possible local Lipschitz constants of $f(x,u)$ with respect to $x$ and $u$.
In Fig. \ref{fig: opt. for estimating Lips}, the objective function of LCQP is represented by a circle. 
We denote its solution as $(\hat L_{x_i}, \hat L_{u_i})$, which corresponds to the points where the circle tangentially intersects the feasible region.
The LCQP problem approximates the possible Lipschitz cone $\mathcal{C}_\text{Lips}$ using a limited number of data points. 
Each data point constitutes a linear constraint that corresponds to a half-plane in the $L_x-L_u$ coordinate plane.
The feasible region of LCQP is then formed by the intersection of these half-planes, providing an estimation of Lipschitz constants based on the available data.

The possible Lipschitz cone is necessarily contained within the feasible region, i.e., $\mathcal{C}_\text{Lips} \subset \mathcal{C}_\text{feas}$. 
As more data becomes available in $\mathcal{I}_i$, the feasible region converges to the possible Lipschitz cone, leading to increasingly accurate estimates of the smallest possible local Lipschitz constants.

\subsection{A simplified expanding algorithm with fixed radius}
The state balls found by MECS are rigorously $\epsilon$-controllable according to one-step controllability backpropagation theorem.
However, its evaluation step suffers from high computational complexity because it requires calculating both the radii of reachable balls and local Lipschitz constants to determine the radius of each expanded ball.
To reduce computational burden, we propose a simplified version of MECS based on shortest path search to identify $\epsilon$-controllable balls. This simplified algorithm is named Floyd expansion with radius fixed (FERF), which maintains the radius of all expanded balls as the same value $\epsilon$ throughout the searching process, neither enlarging it nor reducing it.
This simplification is based on the assumption that two states are mutually controllable if their distance is less than the ball radius.
Under this assumption, if the subsequent state $x_i^\prime$ of a data point $x_i$ is $\epsilon$-controllable, then all states in the ball $\mathcal{B}(x_i,\epsilon)$ are also $\epsilon$-controllable.
Applying this idea to the evaluation step in MECS, we can replace the computation of the expanded ball's radius with a direct assignment, i.e., assign $r_i=\epsilon$ in Algorithm \ref{alg: MECS}.

With this replacement, continuing to use the original four-step iteration of MECS is not the best choice.
A better approach is to convert the $\epsilon$-controllability test to a shortest path problem for a directed graph.
To understand this conversion, let us consider a directed graph $G=\{V, E\}$, where the vertex set $V$ contains all states in the dataset and their subsequent states, plus the target state, i.e.,
$$V=\text{set}\left(\{x_i\}_{i=1}^N\cup\{x_i^\prime\}_{i=1}^N\cup\{x_{\text{T}}\}\right),$$
where the operator $\text{set}(\cdot)$ removes all duplicated elements in a sequence.
The edge set $E$ contains all pairs of vertices $(x_i,x_j)$ such that $x_i$ is controllable with respect to $x_j$ in one step.
Such a one-step controllability holds in two cases.
The first is that $x_j$ is the subsequent state of $x_i$ in the dataset, i.e., $x_j=x_i^\prime$.
The second is that the distance between $x_i$ and $x_j$ is less than or equal to $\epsilon$, so they are mutually controllable according to our assumption.
In the second case, $(x_j,x_i)$ is also in the edge set.
An illustration of such a graph is shown in Fig. \ref{fig: controllability simplified expansion}.
In this graph, whether a state is $\epsilon$-controllable with respect to the target state $x_{\text{T}}$ is equivalent to whether there exists a path from the state to $x_{\text{T}}$.
Therefore, the $\epsilon$-controllability test can be converted to a shortest path problem in the directed graph.
Specifically, for each state in the dataset, we compute the length of its shortest path to the target state.
If and only if the length is finite, the state is $\epsilon$-controllable.
To compute the length of a path, we construct a distance matrix:
$$
A_{i,j}=
\begin{cases}
0, & \text{if } i=j \\
1, & \text{if } (x_i,x_j)\in E \\
\infty, & \text{otherwise},
\end{cases}
$$
where $1\le i,j\le|V|$ and $x_i,x_j$ denote the $i$-th and $j$-th vertices in $V$, respectively.
Matrix $A$ essentially defines the distance between two ordered states as the number of steps from the former to the latter and initializes all distances greater than 1 as infinity.
On the basis of this matrix, we can use Floyd's algorithm~\cite{cormen2022introduction} to compute the distances among all state pairs, i.e., shortest paths in a directed graph.
With all distances computed, for an arbitrarily chosen target state, we can directly obtain all $\epsilon$-controllable states by picking out those with finite distances to the target.
Note that this ability is not achievable by the MECS algorithm since it constructs the $\epsilon$-controllable tree with a given target state as the root node and only finds $\epsilon$-controllable states with respect to this target.

\begin{figure}
  \begin{center}
  \includegraphics[trim=3 2 3 2, clip,width=1\linewidth]{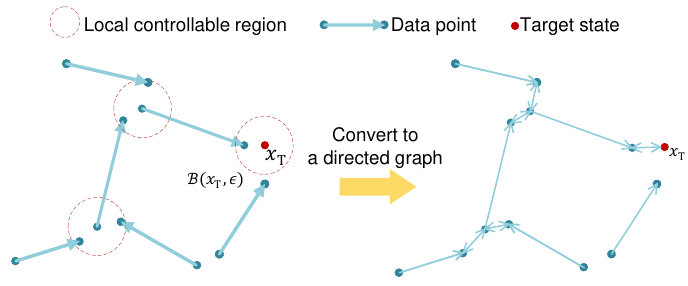}\\
  \caption{A graph search formulation for $\epsilon$-controllability test. If the distance between two states is less than $\epsilon$, they are assumed to be mutually controllable. Under this assumption, the dataset can be converted into a direct graph. Testing $\epsilon$-controllability is equivalent to solving a shortest path problem for this graph.}
  \label{fig: controllability simplified expansion}
  \end{center}
\end{figure}

% === Ⅳ. Complexity Analysis =====================================================
% =================================================================================
\section{Analysis of Algorithm Complexity}
This section analyzes the time complexity of our two algorithms, MECS and FERF, and shows how their computational efficiency can be improved by selecting proper data structures. 
Let us denote the dataset size by $N$, the maximum number of data points in the neighborhood by $n$, and the total number of iterations by $M$.

The time complexity of MECS is analyzed first.
For each sample $(x_i,u_i,x_i^\prime)$, computing its local Lipschitz constant involves two steps: the first is to search for neighboring states of $x_i$ within radius $\delta$; the second is to solve LCQP \eqref{equ: opt. for estimating Lips}.
A naive method for searching neighboring states is to traverse the whole dataset, which has a complexity of $O(N)$.
To alleviate this, we construct a $k$-d tree on the dataset with a complexity of $O(N\log N)$, and query it to obtain neighbors of each state with a complexity of $O(\log N)$ \cite{bentley1975multidimensional}.
The complexity of solving LCQP grows quadratically with respect to its number of inequality constraints \cite{ye1989extension}, which is upper bounded by $n$ in Problem \eqref{equ: opt. for estimating Lips}.
This, solving this LCQP incurs a complexity of $O(n^2)$.
Consequently, the overall complexity of computing Lipschitz constants is $O(N(\log N+n^2))$.

During each iteration, assume that the number of unvisited balls is $S$.
Then, the complexity of finding the one with the maximum radius is $O(S)$, and that of searching for states in the selected ball is $O(\log N)$.
Since the number of states in the selected ball is upper bounded by $n$, the complexity of the pruning step, which traverses both unvisited balls and states in the selected ball, is upper bounded by $O(Sn)$.
Consequently, the complexity of each iteration is $O(\log N+Sn)$.
Since the number of unvisited balls $S$ varies across different iterations, we set a loose upper bound $N$ on it.
Therefore, the complexity of each iteration is upper bounded by $O(Nn)$, and the overall complexity of MECS is upper bounded by $O(MNn)$.
In general cases, the total number of iterations $M$ is difficult to determine.
As a special case, when the Lipschitz constant $L_x$ is greater than $1$, we always have $M \leq N$, and the complexity is upper bounded by $O(N^2n)$.
This is because in this case, the radius of expanded ball gradually decreases, ensuring that the same state will not be expanded twice.
In practice, the $k$-d tree greatly improves the computational efficiency of MECS.
However, in complexity analysis, since the number of unvisited balls $S$ is relaxed to $N$, the acceleration of $k$-d tree is not reflected in the upper bound of overall time complexity.

For FERF, we assume that the size of vertex set is $L$.
Due to the inherent characteristics of the directed graph, we have the inequality $L\leq2N+1$.
The time complexity of filling the vertex set and initializing the distance matrix is $O(N)$, and the complexity of filling the distance matrix using Lipschitz continuity assumption is $O(L^2)$.
Finally, the complexity of solving the shortest path problem using Floyd's algorithm is $O(L^3)$.
Therefore, the overall time complexity of FERF is $O(N^3)$.
This complexity does not include an indeterminable number of iterations $M$ as MECS because Floyd's algorithm is guaranteed to end after a three-level traversal of the vertex set.
Moreover, FERF finds $\epsilon$-controllable states with respect to all possible target states instead of just a given one as MECS.
For the case of a given target state, we can use another shortest path algorithm, e.g., Dijkstra's algorithm, to compute the distances of all states to a specific target.
In this case, the complexity of FERF can be further reduced to $O(N^2)$.

% === Ⅴ. Experiments =============================================================
% =================================================================================
\section{Experimental Validation}
This section evaluates our controllability test method in three datatic control systems, including both linear and nonlinear ones.
To measure how many states can be transferred to a given target, we introduce a new performance index called degree of controllability (DOC), which is defined as the proportion of $\epsilon$-controllable states to all states in the dataset.
For each system, we focus on four aspects: 1) the process of $\epsilon$-controllability expansion and the final $\epsilon$-controllable subsets; 2) the impact of different error radii on DOC for the same target state; 3) the impact of different target states on DOC for the same error radius; and 4) comparison of estimated Lipschitz constant with the true value.

\subsection{Mass-spring system}
We start validating our MECS algorithm in a linear time-invariant system. 
Consider a 2-dimensional mass-spring system:
\begin{equation}
\begin{aligned}
\begin{bmatrix}
    y^{\prime} & \dot{y}^{\prime}
\end{bmatrix}^\top&=
A\begin{bmatrix}
    y & \dot{y}
\end{bmatrix}^\top+Bu, \\
A&=\begin{bmatrix}
    1 & \Delta t\\
    -\frac{k}{m}\Delta t & 1-\frac{\rho}{m}\Delta t
\end{bmatrix}, \\
B&=\begin{bmatrix}
    0 & \frac{1}{m}\Delta t
\end{bmatrix}^\top,
\label{massspring}
\end{aligned}
\end{equation}
where $y$ is the displacement, $m=\SI{0.5}{\kilo\gram}$ is the mass, $k=\SI{1}{\newton\per\metre}$ is the elastic coefficient, $\rho=\SI{1.5}{\newton\per\metre\per\second}$ is the damping coefficient, and $\Delta t=\SI{0.1}{\second}$ is the sampling time.
The system state is $x=[y,\dot{y}]^\top$ and the equilibrium is $x_{\mathrm{equ}}=[0.0,0.0]$. 
We specify a bounded working space $\{\mathcal{X}, \mathcal{U}\}$ for data collection and controllability verification:
\begin{equation}
\begin{split}
&\mathcal{X}=[-1.0,1.0]\times[-1.0,1.0],\\
&\mathcal{U}=[-1.0,1.0].
\label{bound for spring}
\end{split}
\end{equation}

This linear system is controllable because its controllability matrix has full row rank, i.e.
\begin{equation}
    \text{rank}([\begin{matrix}
        A&AB
    \end{matrix}]) = 2.
\end{equation}
It means that for any two states, there exists a series of control inputs that steer one state to the other.

We first collect the dataset consisting of multiple state-action trajectories.
This dataset, shown in Fig. \ref{fig: massspring sampling a}, exhibits dense clustering of states around the equilibrium point.
This clustering is attributed to the system's tendency to transfer toward the equilibrium point due to damping.
The initial state is randomly chosen in the state space, and actions are determined by a random policy.
Once the state-action trajectory reaches a maximum length of $50$, the system state is reset.
The sampling process continues until the dataset contains a total of $N = 5000$ state transition pairs.
Data collection for subsequent experiments followed a similar process.

For a linear system, the Lipschitz constant is the same for all states and can be theoretically computed as
\begin{equation}
L_x^* = \sup_{x_1,x_2\in\mathcal{X}}\frac{\lVert (Ax_1+Bu)-(Ax_2+Bu)\rVert_2}{\lVert x_1 - x_2 \rVert_2}=\lVert A \rVert_2.
\end{equation}
The true Lipschitz constant is the Euclidean norm of system matrix, i.e., $L_x^* = \lVert A \rVert_2 = 1.021$.
By solving LCQP with a confidence radius of $\delta = 0.2$, the estimate of local Lipschitz constant is $\hat L_x=0.98$.
The small error in Lipschitz constant estimation is acceptable because it does not have much effect on the size of $\epsilon$-controllable subsets found by MECS.

After data collection and Lipschitz constant estimation, we start to identify $\epsilon$-controllable states using MECS.
With the error radius $\epsilon=0.05$ and the target state $x_{\text{T}}=x_\text{equ}$, the final $\epsilon$-controllable subsets are visualized in Fig. \ref{fig: massspring sampling}.
The $\epsilon$-controllable subsets identified in different iterations are shown in Fig. \ref{fig: mass_spring_expansion}.
These subsets gradually expand from the neighborhood of equilibrium point until covering almost all states in the dataset.

\begin{figure}[!htbp]
  \begin{center}
  \subfloat[Data points]{
  \includegraphics[trim=0 5 30 30, clip, width=0.48\linewidth]{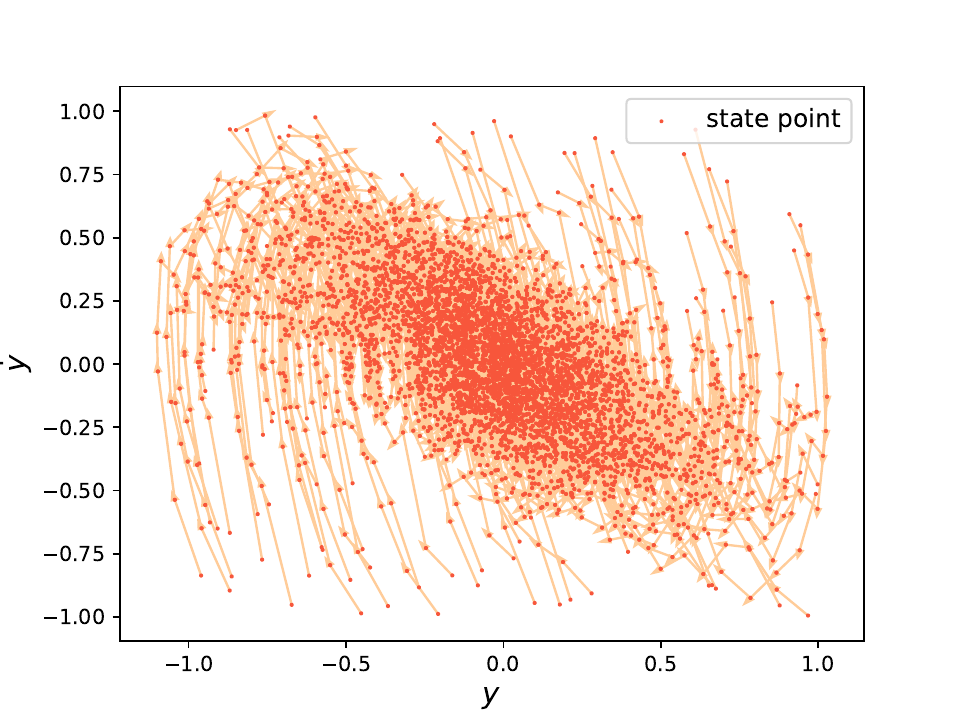}
  \label{fig: massspring sampling a}
  }
  \subfloat[$\epsilon$-controllable subset]{
  \includegraphics[trim=0 5 30 30, clip, width=0.48\linewidth]{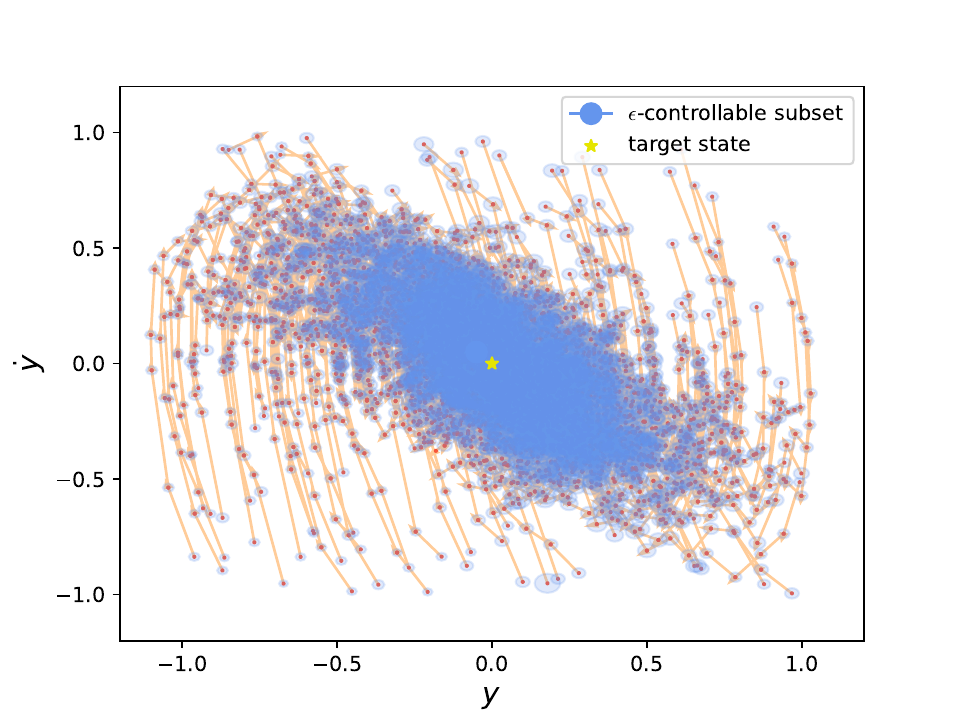}
  \label{fig: massspring sampling b}
  }
  \caption{Sampled data points and identified $\epsilon$-controllable subsets in a mass-spring system when $\epsilon=0.05$. The red points represent states, and the orange arrows represent their corresponding time derivatives. The length of the arrows represents the norm of time derivatives. The blue area represents $\epsilon$-controllable subsets.}
  \label{fig: massspring sampling}
  \end{center}
\end{figure}

\begin{figure}
    \subfloat[Step=$1$]{
        \includegraphics[trim=0 5 30 30, clip, width=0.48\linewidth]{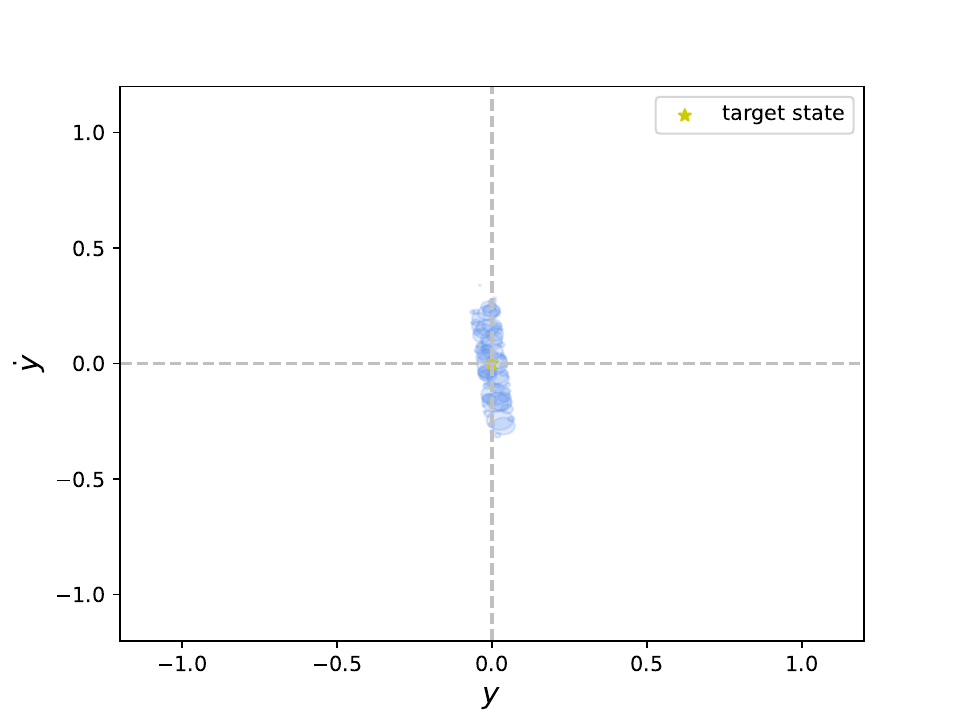}
    }
    \subfloat[Step=$1500$]{
        \includegraphics[trim=0 5 30 30, clip, width=0.48\linewidth]{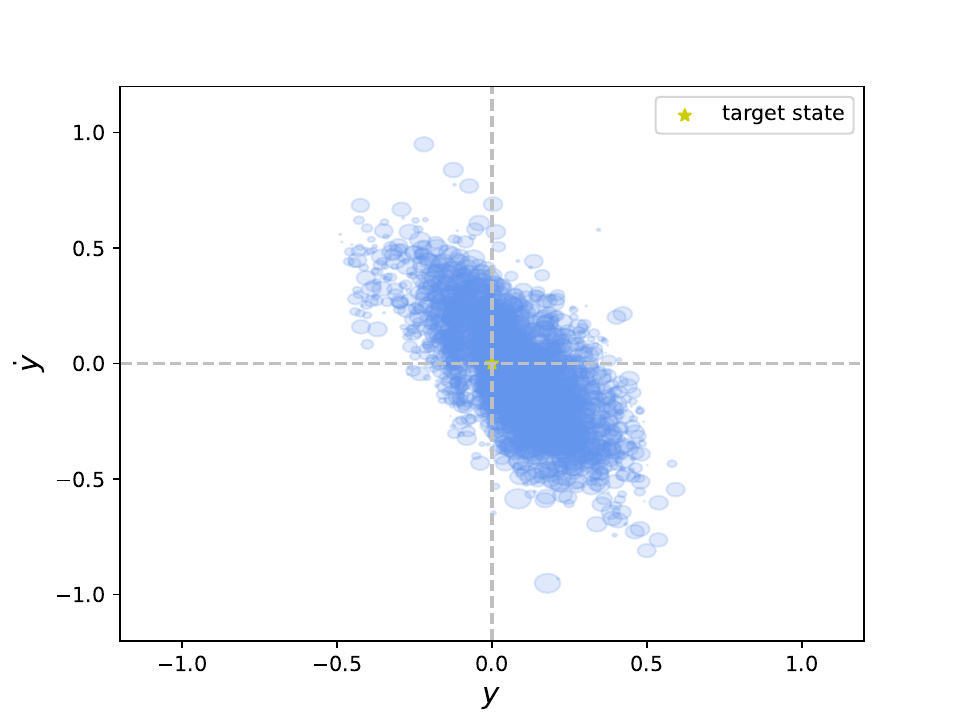}
    }\\
    \subfloat[Step=$3000$]{
        \includegraphics[trim=0 5 30 30, clip, width=0.48\linewidth]{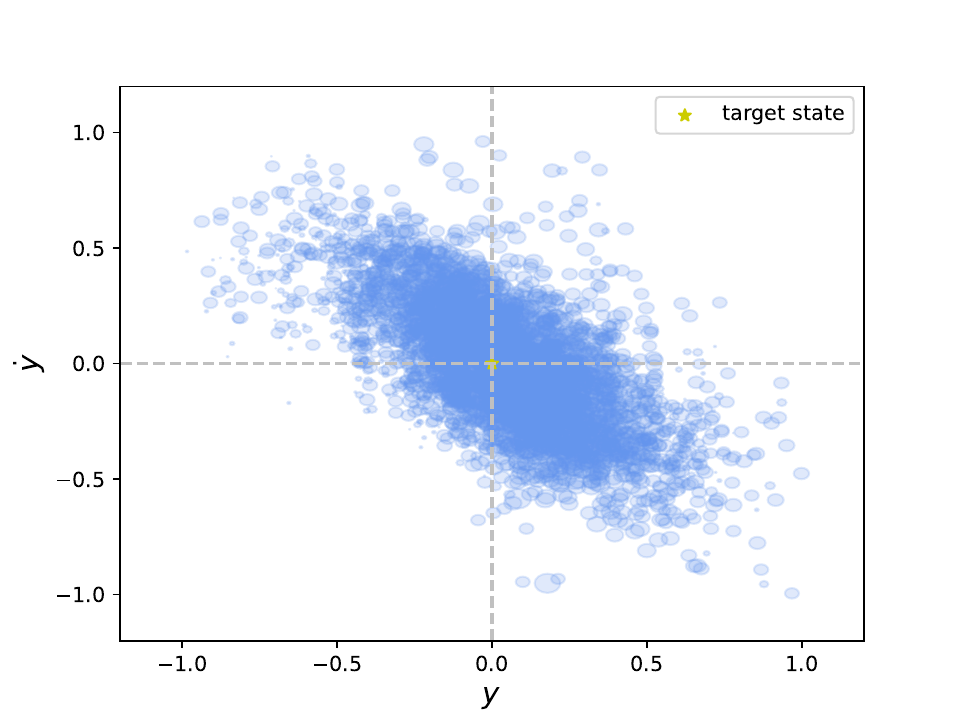}
    }
    \subfloat[Step=$4500$]{
        \includegraphics[trim=0 5 30 30, clip, width=0.48\linewidth]{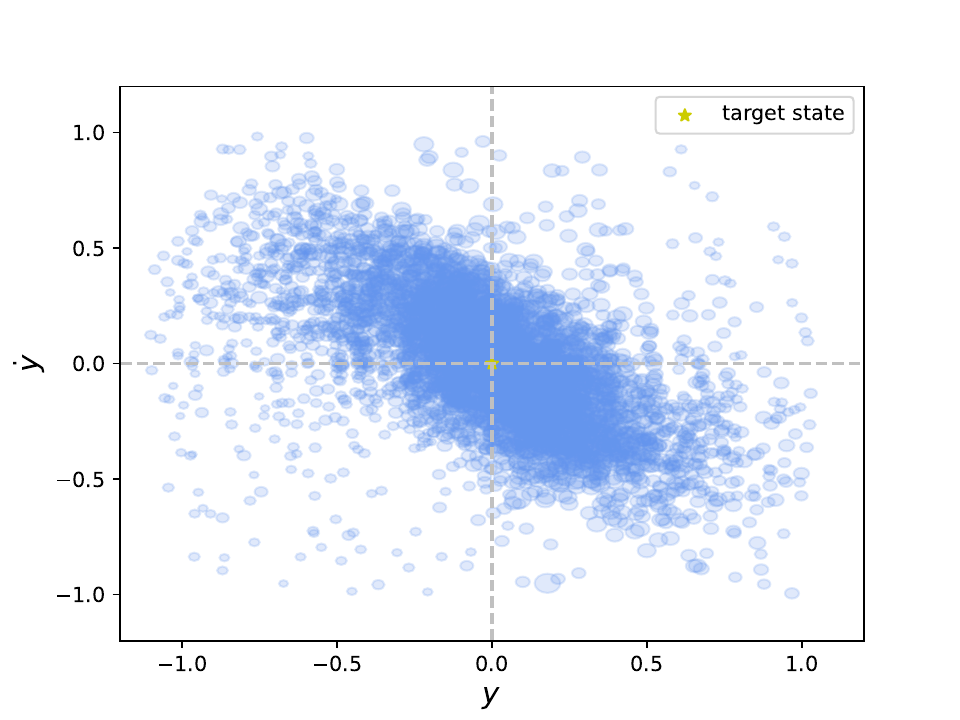}
    }
    \caption{The $\epsilon$-controllable subsets expansion process for a mass-spring system when equilibrium point is targeted $(\epsilon=0.05)$. The blue area represents $\epsilon$-controllable subsets.}
    \label{fig: mass_spring_expansion}
\end{figure}

The relationship between the error radius $\epsilon$ and DOC is visualized in Fig. \ref{fig: massspring epsilon a}. 
As $\epsilon$ increases, the interested target ball is enlarged, leading to an increased proportion of $\epsilon$-controllable states.
When the target state is set to the equilibrium point, all states in the dataset are $\epsilon$-controllable as long as $\epsilon$ exceeds $0.02$.
When the target state deviates from the equilibrium point, a larger $\epsilon$ is required to ensure $\epsilon$-controllability of all states.
When $\epsilon$ becomes sufficiently large, the target ball covers the entire dataset, and all states become $\epsilon$-controllable.

\begin{figure}
  \begin{center}
  \subfloat[DOC changes with $\epsilon$]{\includegraphics[trim=10 5 15 10, clip, width=0.48\linewidth]{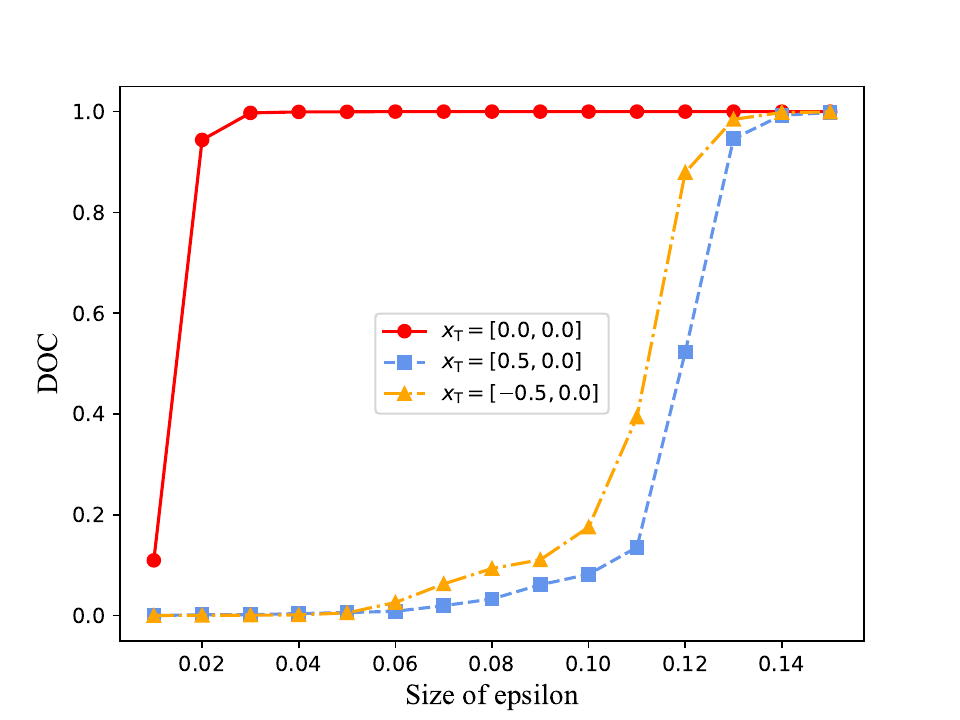}
  \label{fig: massspring epsilon a}
  }
  \subfloat[DOC for different target states]{\includegraphics[trim=20 30 35 30, clip, width=0.48\linewidth]{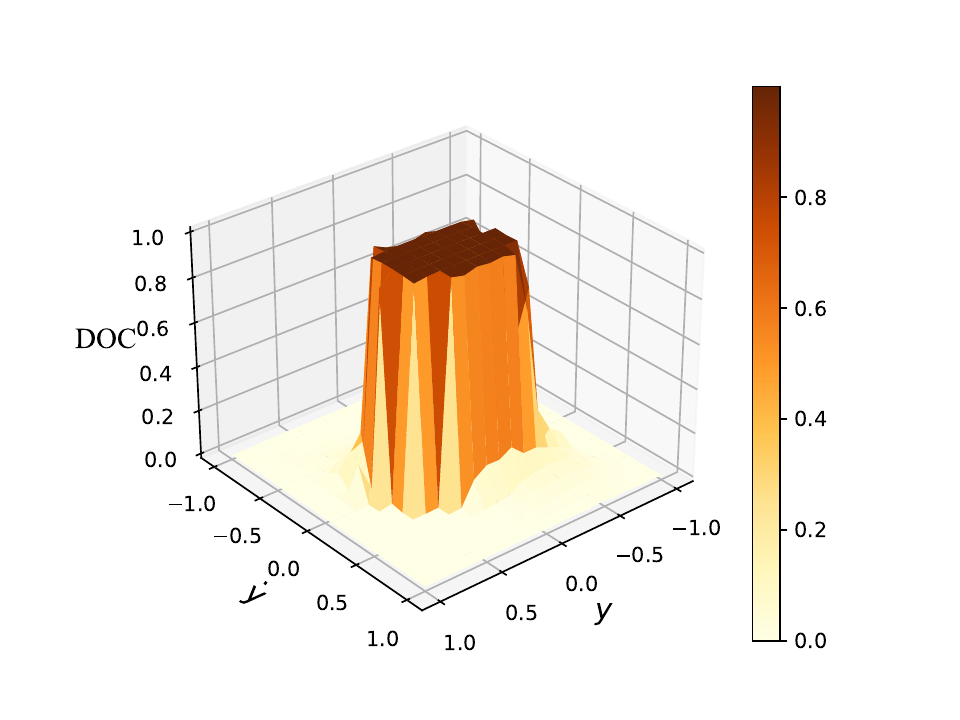}
  \label{fig: massspring epsilon b}
  }
  \caption{The effect of different error radii and target states on DOC in a mass-spring system.}
  \label{fig: massspring epsilon}
  \end{center}
\end{figure}
We calculate DOC for various target states, with a fixed error radius $\epsilon=0.05$, as shown in Fig. \ref{fig: massspring epsilon b}.
It shows that DOC is high in the region near the equilibrium point.
Beyond a certain distance, DOC decreases dramatically as the target state moves away from the equilibrium point.
This is because states near the equilibrium point are more likely to be reached, as system energy tends to dissipate due to friction.

Next, we examine a system that is uncontrollable in the sense of traditional exact controllability.
We discard the control input in the mass-spring system where the system matrix $A$ has two real eigenvalues less than $1$.
The dataset and identified $\epsilon$-controllable subset with respect to the equilibrium point are shown in Fig. \ref{fig: massspring wocontrol sampling}.
As all states follow their specific trajectories towards the equilibrium point, they are all $\epsilon$-controllable when the target state is set to the equilibrium point.

\begin{figure}[!htbp]
  \begin{center}
  \subfloat[Data points]{\includegraphics[trim=0 5 30 30, clip, width=0.48\linewidth]{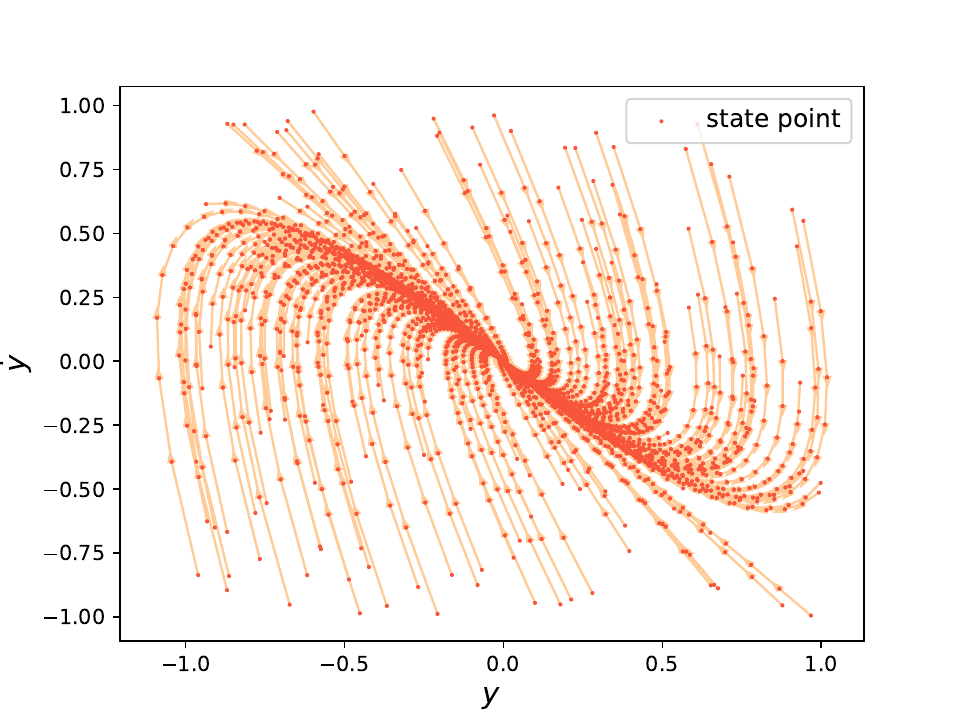}}
  \subfloat[$\epsilon$-controllable subset]{\includegraphics[trim=0 5 30 30, clip, width=0.48\linewidth]{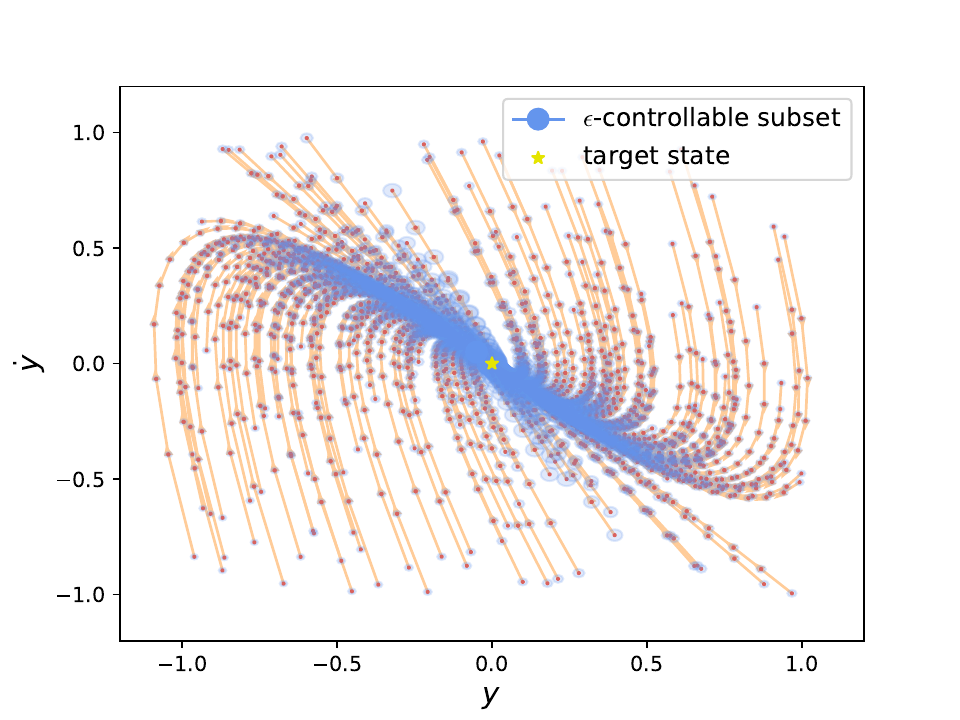}}
  \caption{Sampled data points and identified $\epsilon$-controllable subsets in a mass-spring system without control input when $\epsilon=0.05$.}
  \label{fig: massspring wocontrol sampling}
  \end{center}
\end{figure}

We further analyze the relationship between the error radius $\epsilon$ and DOC under different target states, as shown in Fig. \ref{fig: massspring wocontrol epsilon a}.
When the target state is set to the equilibrium point, all states are controllable even with a small $\epsilon$. 
When the target state deviates from the equilibrium point, DOC remains small until $\epsilon$ exceeds $0.16$.
We also calculate DOC for different targets in the state space with $\epsilon=0.05$, as visualized in Fig. \ref{fig: massspring wocontrol epsilon b}.
It shows that DOC remains low throughout the state space except for a tiny region near the equilibrium point.
This phenomenon arises because, in a system without control input, the states are steered to the equilibrium point along their particular trajectories but lack the ability to transfer between each other.

\begin{figure}
  \begin{center}
  \subfloat[DOC changes with $\epsilon$]{\includegraphics[trim=10 5 15 10, clip, width=0.48\linewidth]{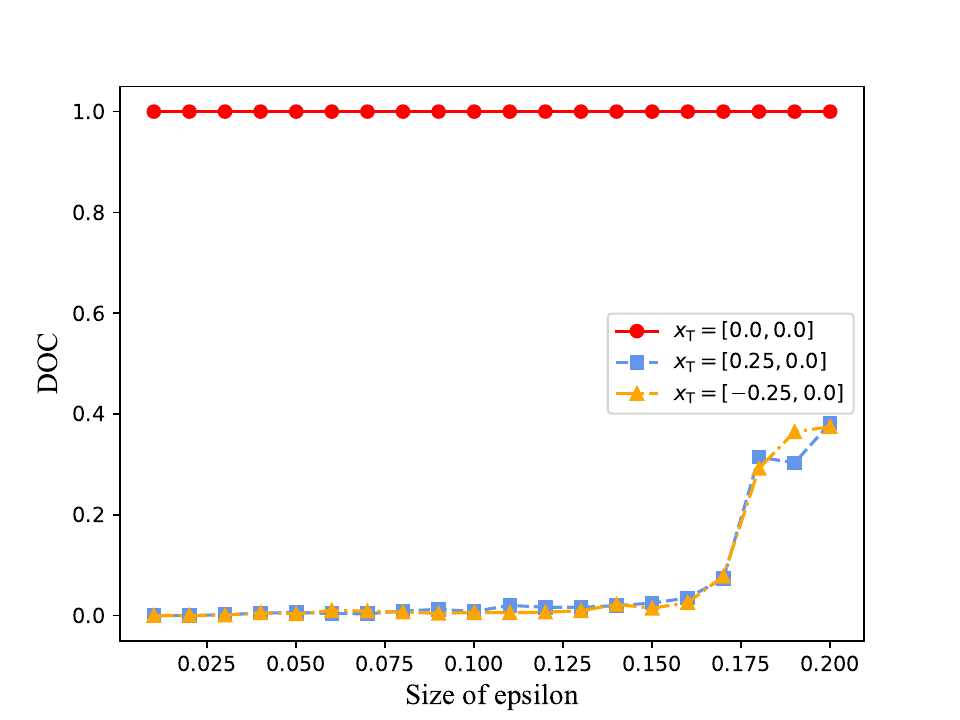}\label{fig: massspring wocontrol epsilon a}}
  \subfloat[DOC for different target states]{\includegraphics[trim=20 30 35 30, clip, width=0.48\linewidth]{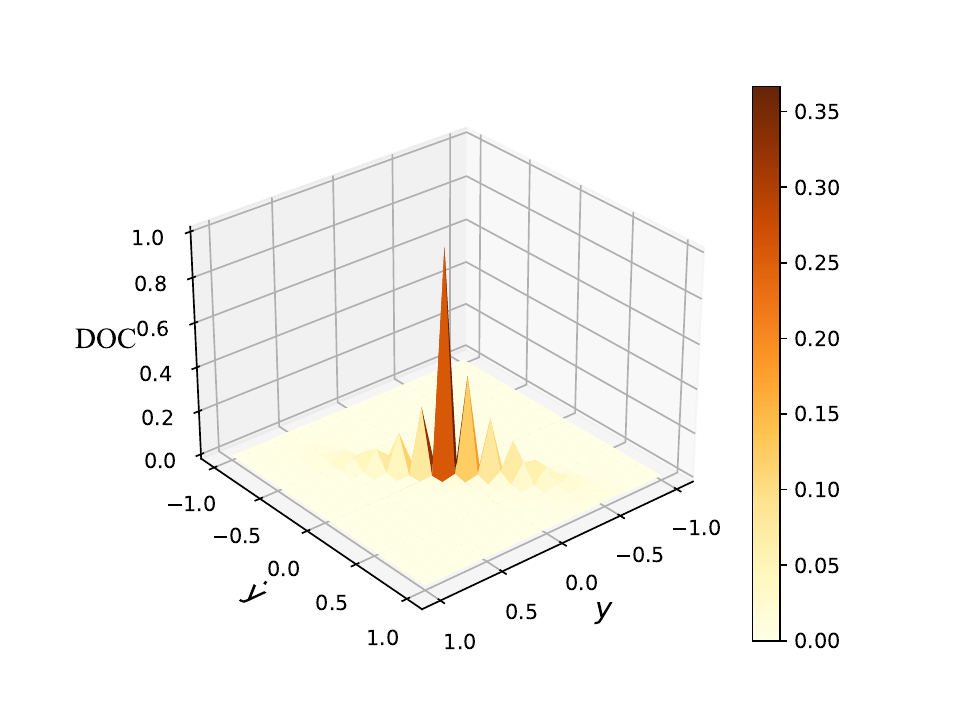}\label{fig: massspring wocontrol epsilon b}}
  \caption{The effect of different error radii and target states on DOC in a mass-spring system without control input.}
  \label{fig: massspring wocontrol epsilon}
  \end{center}
\end{figure}

\subsection{Oscillator}
We consider a Van der Pol oscillator, which is a highly nonlinear system:
\begin{equation}
\begin{split}
    y^{\prime} &= y + \dot{y}\Delta t,\\
    \dot{y}^{\prime} &=-y\Delta t+\dot{y}(1-\frac{1}{2}(1-y^2))\Delta t+u\Delta t,
\end{split}
\label{eq: oscillator dynamics}
\end{equation}
where $y$ is the position coordinate, $x=[y,\dot{y}]^\top$ is the system state, and $\Delta t=\SI{0.1}{\second}$ is the sampling time.
The equilibrium point is $x_{\text{equ}}=[0.0,0.0]$. We specify a bounded state-action space for data collection and controllability test:
\begin{equation}
\begin{split}
&\mathcal{X}=[-1.0,1.0]\times[-1.0,1.0],\\
&\mathcal{U}=[-0.5,0.5].
\label{bound for oscillator}
\end{split}
\end{equation}

Data collection is conducted using the aforementioned sampling method with a maximum trajectory length of $200$.
The sampled data points are shown in Fig. \ref{fig: oscillator sampling a}, in which the number of data points is $N=5000$.
Then, the local Lipschitz constants for each state-action pair are estimated by solving LCQP.
We visualize LCQP at $x_{\text{equ}}$ in Fig. \ref{fig: oc_opt}. 
It shows that our method accurately estimates the local Lipschitz constants when data is sufficient.

\begin{figure}[!htbp]
  \begin{center}
  \includegraphics[trim=10 10 10 10, clip, width=0.6\linewidth]{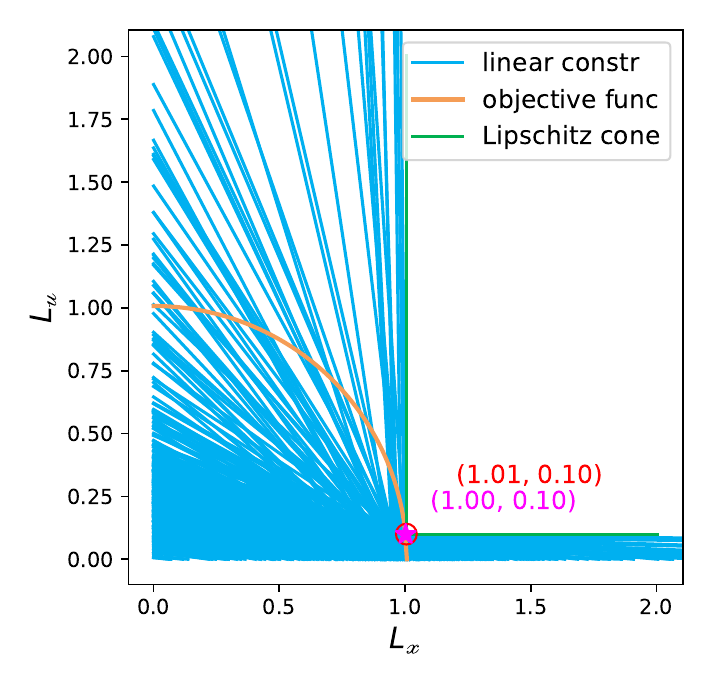}
  \caption{The estimation of local Lipschitz constant in an oscillator system. In the $L_x-L_u$ coordinate plane, the region enclosed by the green line represents all possible Lipschitz constants. $L_{x_i}^*$ and $L_{u_i}^*$ denote the smallest possible local Lipschitz constants of $f(x,u)$ with respect to $x$ and $u$. The blue lines represent the linear constraints of LCQP. The orange curve depicts the objective function of LCQP. In this illustration, $(\hat L_{x_i},\hat L_{u_i})$ is the optimal solution to LCQP.}
  \label{fig: oc_opt}
  \end{center}
\end{figure}

For controllability test, we use MECS to find all $\epsilon$-controllable states. 
To better understand how the target state affects controllability, we set it to a non-equilibrium point $x_{\text{T}}=[0.25,0]$.
The identified $\epsilon$-controllable subsets are visualized in Fig. \ref{fig: oscillator sampling b}. 
The result is consistent with the intuition that nearly all states in the dataset are $\epsilon$-controllable.

\begin{figure}[!htbp]
  \begin{center}
  \subfloat[Data points]{\includegraphics[trim=0 5 30 30, clip, width=0.48\linewidth]{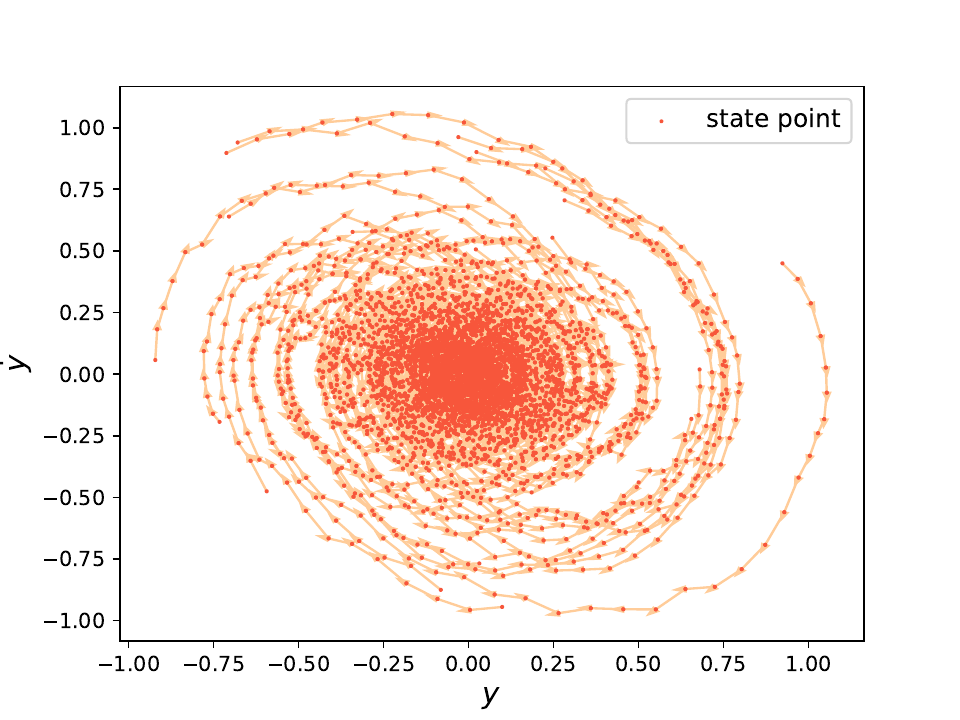}\label{fig: oscillator sampling a}}
  \subfloat[$\epsilon$-controllable subset]{\includegraphics[trim=0 5 30 30, clip, width=0.48\linewidth]{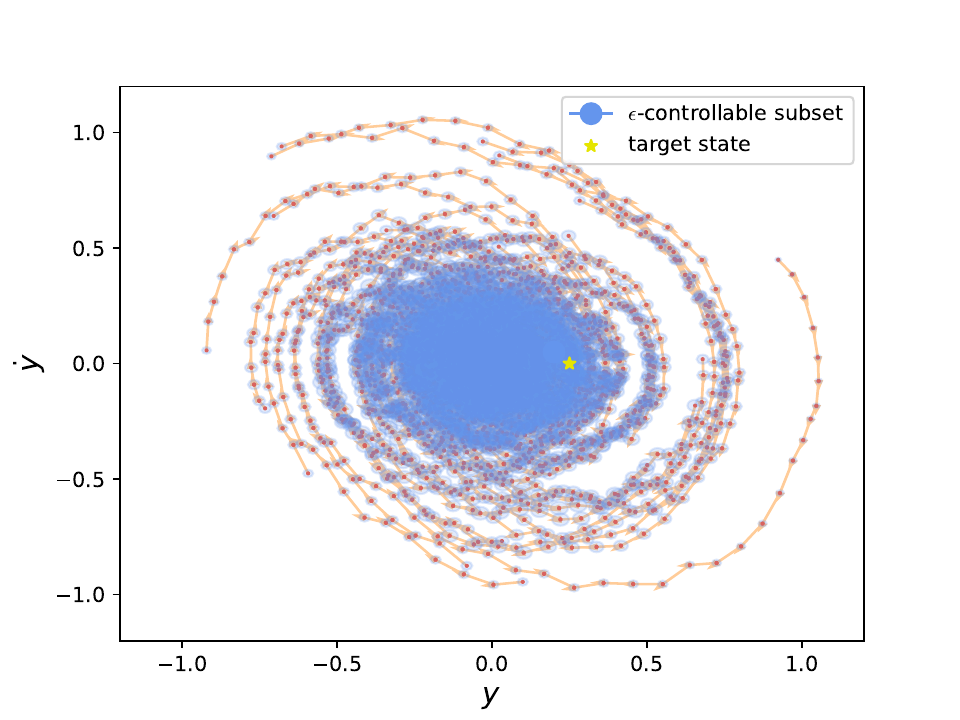}\label{fig: oscillator sampling b}}
  \caption{Sampled data points and identified $\epsilon$-controllable subsets in an oscillator system when $\epsilon=0.05$.}
  \label{fig: oscillator sampling}
  \end{center}
\end{figure}

\begin{figure}
    \subfloat[Step=$1$]{
        \includegraphics[width=0.48\linewidth]{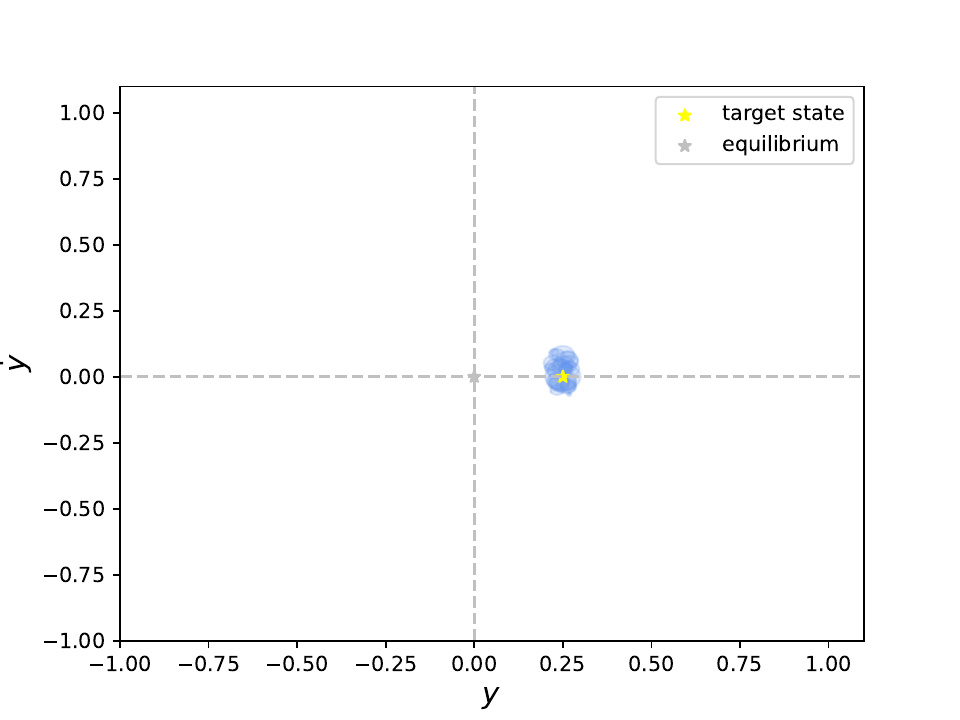}
    }
    \subfloat[Step=$1500$]{
        \includegraphics[width=0.48\linewidth]{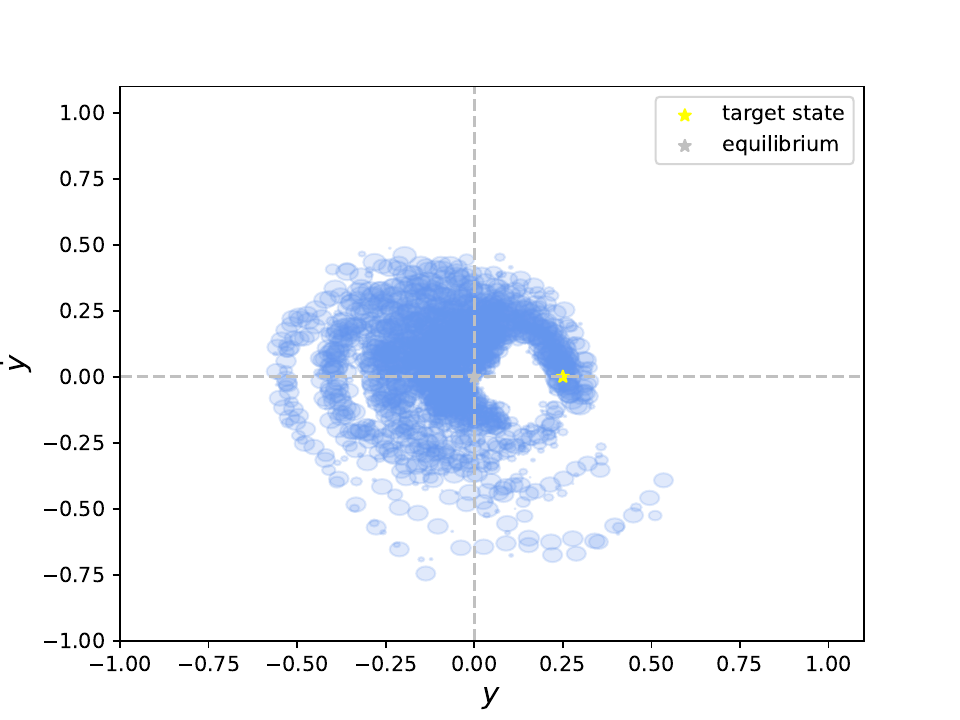}
    }\\
    \subfloat[Step=$3000$]{
        \includegraphics[width=0.48\linewidth]{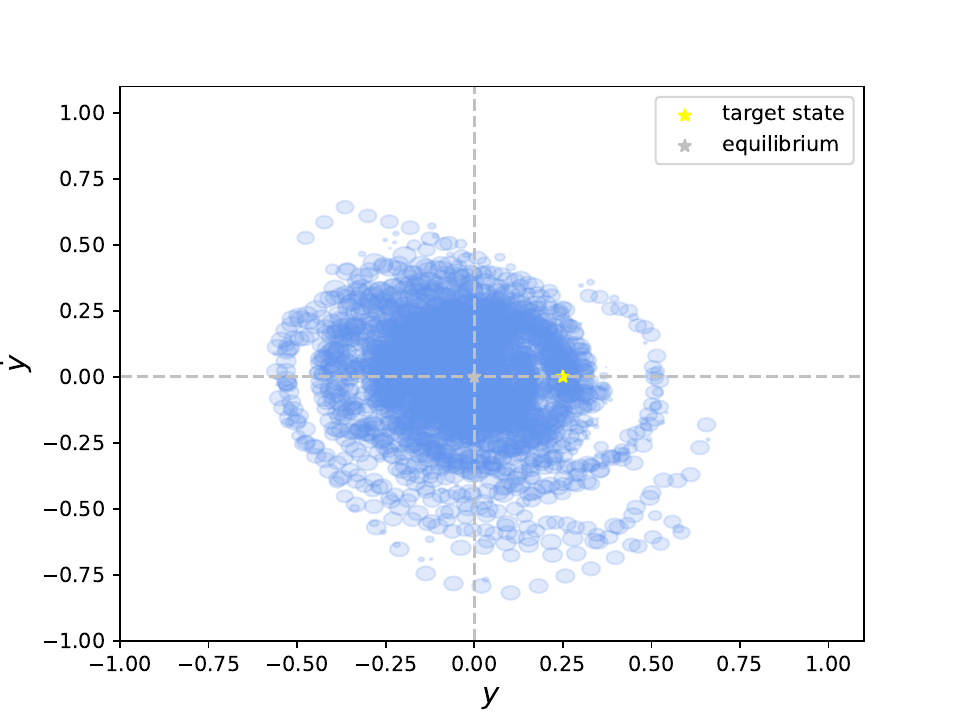}
    }
    \subfloat[Step=$4500$]{
        \includegraphics[width=0.48\linewidth]{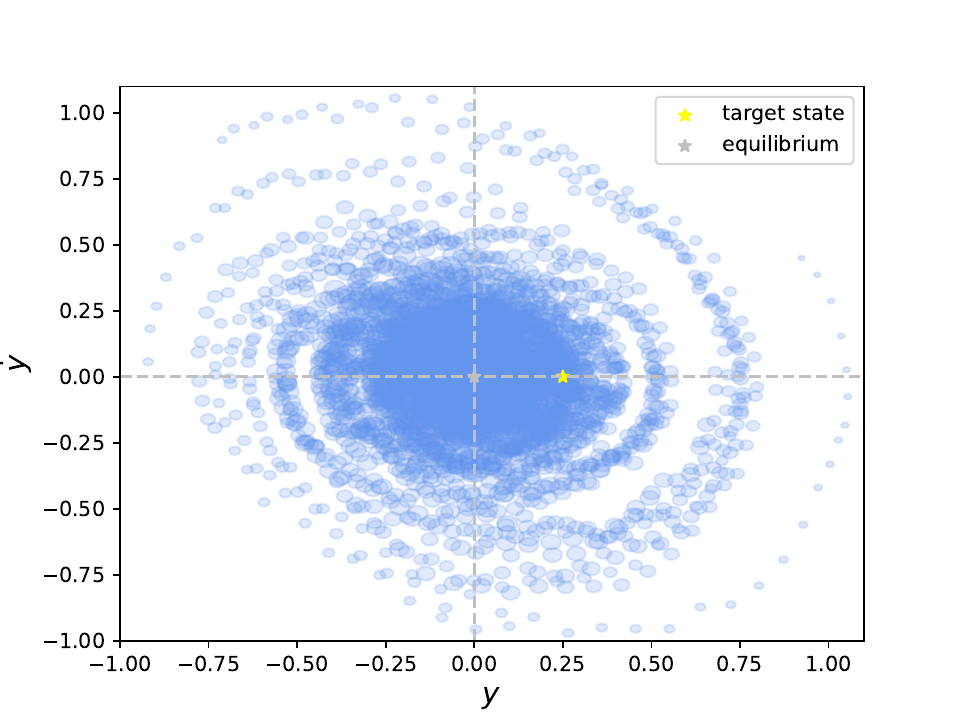}
    }
    \caption{The process of $\epsilon$-controllable set expansion for controllability verification in an oscillator system $(\epsilon=0.05)$, where the blue area represents the $\epsilon$-controllable subset.}
    \label{fig: oc expansion}
\end{figure}

We examine the influence of the error radius $\epsilon$ on controllability, as shown in Fig. \ref{fig: oscillator control epsilon a}.
It indicates that for $\epsilon\geq0.02$, all states in the dataset are $\epsilon$-controllable with respect to the equilibrium point.
When target states deviate from the equilibrium point, DOC surges only when $\epsilon$ is larger than $0.06$. 
This indicates that the equilibrium point exhibits superior controllability compared to its surrounding states in the sense of DOC.
Fig. \ref{fig: oscillator control epsilon b} shows DOC for different target states when $\epsilon=0.05$.
In a region near the equilibrium point, almost all states are $\epsilon$-controllable, while outside this region, DOC quickly drops to zero.

\begin{figure}
  \begin{center}
  \subfloat[DOC changes with $\epsilon$]{\includegraphics[trim=10 5 15 10, clip, width=0.48\linewidth]{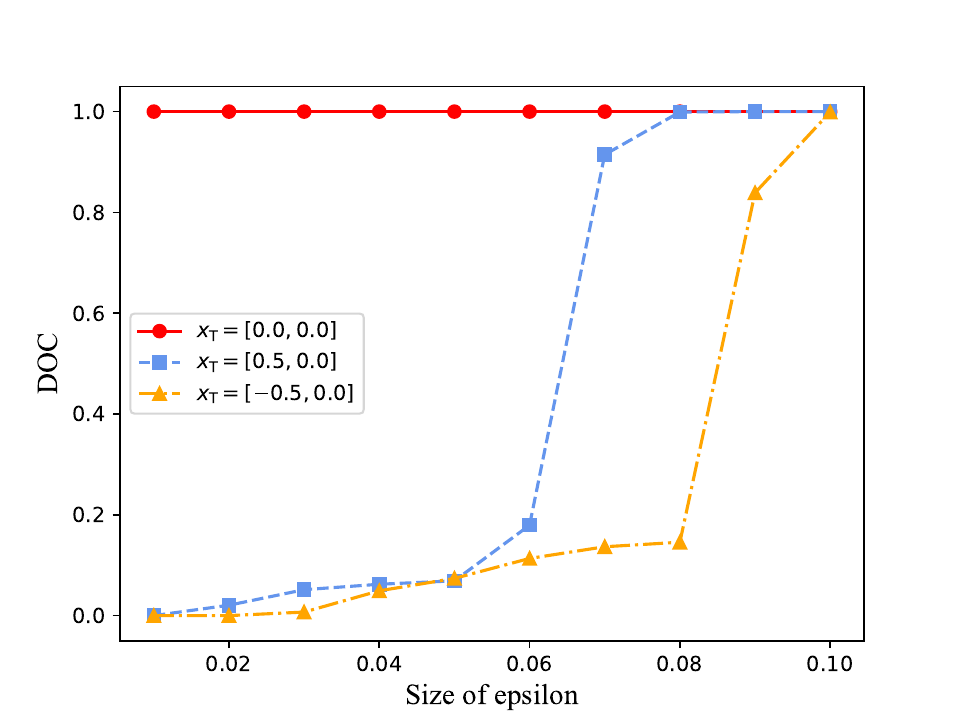}
  \label{fig: oscillator control epsilon a}}
  \subfloat[DOC for different target states]{\includegraphics[trim=20 30 35 30, clip, width=0.48\linewidth]{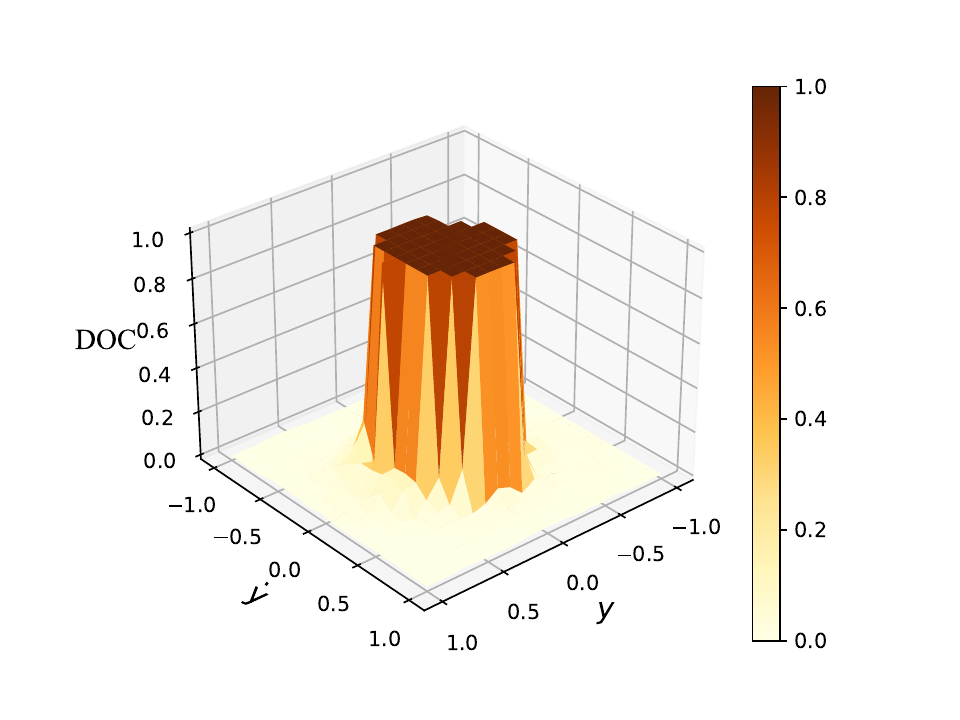}
  \label{fig: oscillator control epsilon b}}
  \caption{Relationship between changes in DOC with error radius $\epsilon$ and target state in an oscillator system.}
  \end{center}
\end{figure}

We consider an uncontrollable version of the oscillator system by removing its control input.
The sampled data points are shown in Fig. \ref{fig: oscillator wo sampling a}.
The states are steered towards the equilibrium point following helical trajectories.
Fig. \ref{fig: oscillator wo sampling b} shows the identified $\epsilon$-controllable subsets when $x_{\text{T}}$ deviates from $x_{\text{equ}}$.
Due to the lack of control input, only states on specific helical curves can be steered to the target state. 

\begin{figure}[!htbp]
  \begin{center}
  \subfloat[Data points]{\includegraphics[trim=0 5 30 30, clip, width=0.48\linewidth]{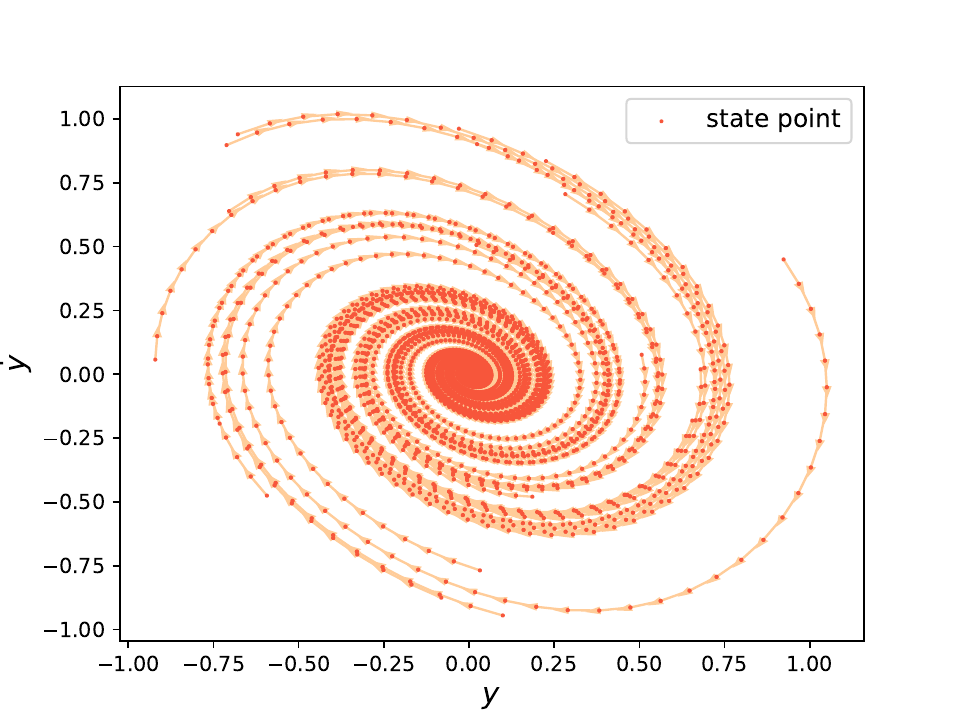}\label{fig: oscillator wo sampling a}}
  \subfloat[$\epsilon$-controllable subset]{\includegraphics[trim=0 5 30 30, clip, width=0.48\linewidth]{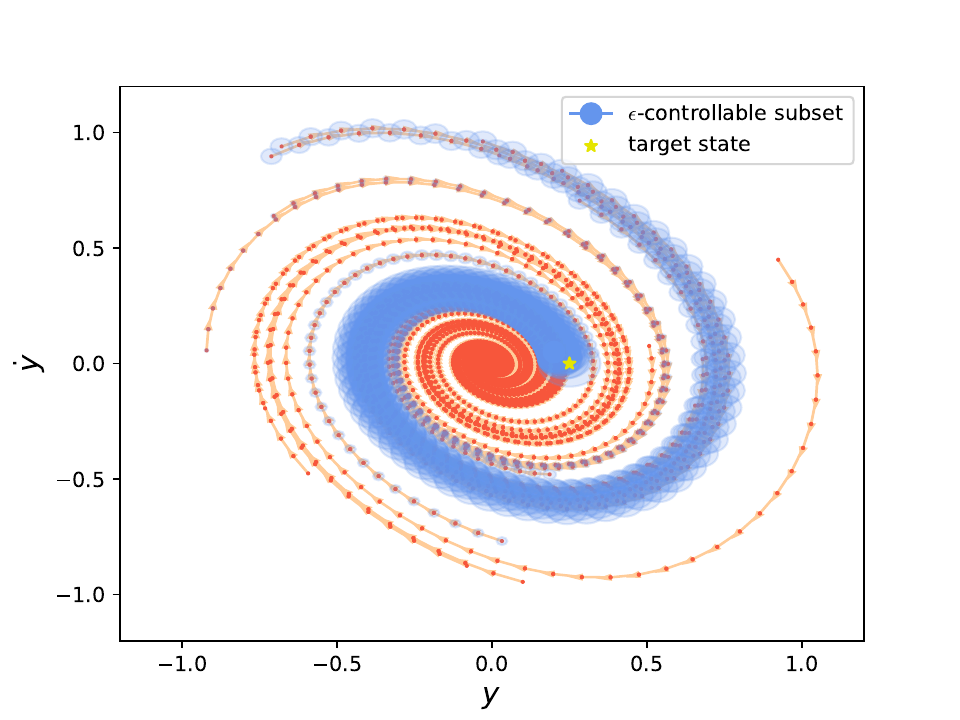}\label{fig: oscillator wo sampling b}}
  \caption{Sampled data points and identified $\epsilon$-controllable subsets in an oscillator system without control input when $\epsilon=0.05$.}
  \label{fig: oscillator wo sampling}
  \end{center}
\end{figure}

Fig. \ref{fig: oscillator wocontrol epsilon a} shows DOC under different error radii.
For the equilibrium point, a small $\epsilon$ is sufficient to make all states $\epsilon$-controllable.
For non-equilibrium points, DOC significantly drops when $\epsilon$ is small and slowly rises as $\epsilon$ increases.
The relationship between DOC and the target state is visualized in Fig. \ref{fig: oscillator wocontrol epsilon b}.
It clearly shows that almost all states are $\epsilon$-controllable when the target state is set to the equilibrium point.
When the target state deviates from the equilibrium point, DOC drops much more quickly to zero than in the case with control input.

\begin{figure}
  \begin{center}
  \subfloat[DOC changes with $\epsilon$]{\includegraphics[trim=10 5 15 10, clip, width=0.48\linewidth]{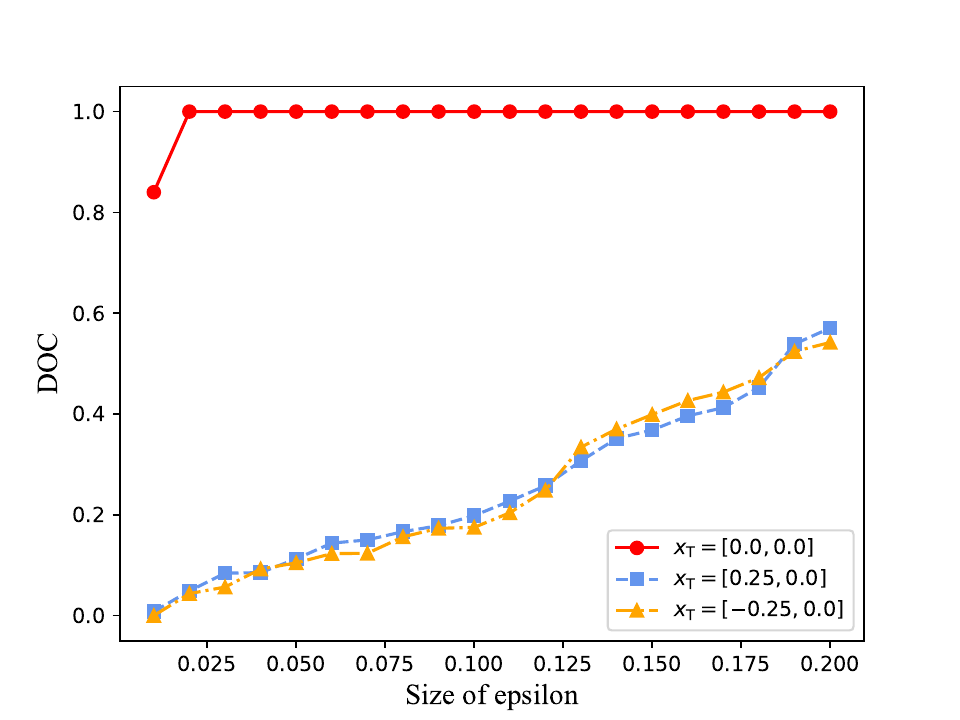}\label{fig: oscillator wocontrol epsilon a}}
  \subfloat[DOC for different target states]{\includegraphics[trim=20 30 35 30, clip, width=0.48\linewidth]{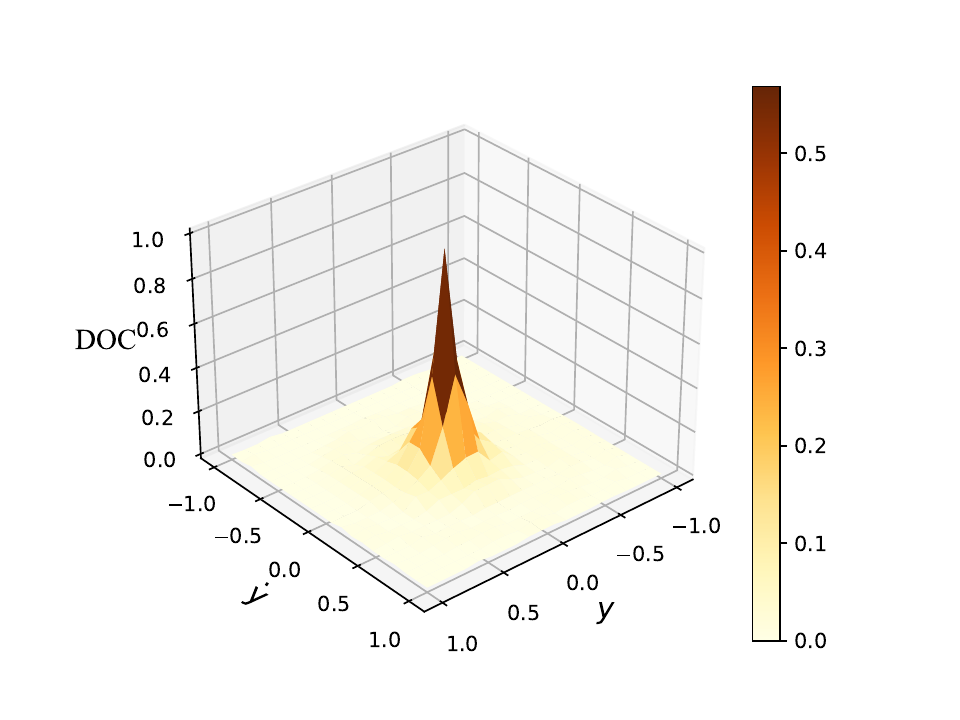}\label{fig: oscillator wocontrol epsilon b}}
  \caption{The effect of different error radii and target states on DOC in an oscillator system without control input.}
  \end{center}
\end{figure}

\subsection{Tunnel-diode circuit}
In the aforementioned systems, there is only one equilibrium point, and all states in the dataset are $\epsilon$-controllable with respect to it.
Next, we explore a more complex system with multiple equilibrium points.
Consider a tunnel-diode circuit:
\begin{equation}
\begin{aligned}
    {{x}_1}^{\prime}&=x_1+\frac1C(-h(x_1)+x_2)\Delta t,\\
    {{x}_2}^{\prime}&=x_2 +\frac1L(-x_1-Rx_2+u)\Delta t,
\end{aligned}
\end{equation}
whose parameters are $R=\SI{1.5e3}{\ohm}, C=\SI{2e-12}{\farad}$, $L=\SI{5e-6}{\henry}$, and $\Delta t=\SI{0.1}{\second}$, and $h(\cdot)$ is given by
\begin{equation}
\begin{split}
    h(x_1)&=17.76x_1-103.79x_1^2+229.62x_1^3\\&
    -226.31x_1^4+83.72x_1^5.
\end{split}
\end{equation}
The control input is constant at $u=\SI{1.2}{\volt}$.
We specify a bounded state space for data collection and controllability verification: 
\begin{equation}
\mathcal{X}=[-0.3,1.4]\times\left[-0.3,1.4\right].
\end{equation}

\begin{figure}[!htbp]
  \begin{center}
  \includegraphics[trim=0 5 30 30, clip, width=0.6\linewidth]{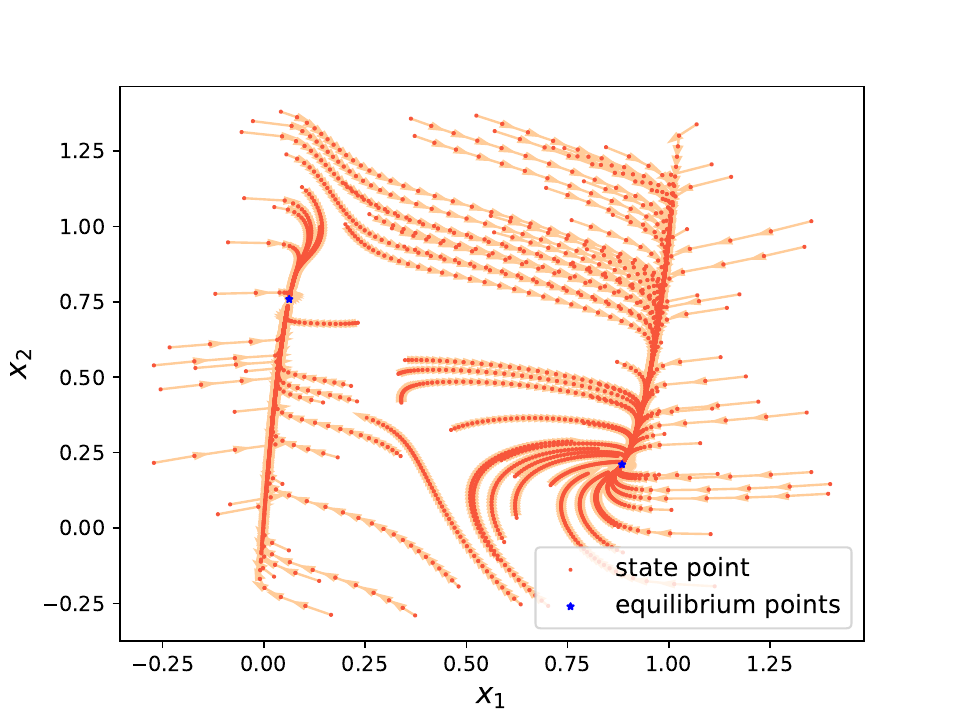}\\
  \caption{Data points for controllability verification in a tunnel-diode circuit.}
  \label{fig: tunnel sampling}
  \end{center}
\end{figure}

The collected data points are shown in Fig. \ref{fig: tunnel sampling}, in which the number of data points is $N=5000$.
We can verify that the system has two equilibrium points.
In fact, by setting $\dot{x_1}=\dot{x_2}=0$, we have three equilibrium points: $x_{\mathrm{equ0}}=[0.285.0.61]$, $x_{\mathrm{equ1}}=[0.063.0.758]$, and $x_{\mathrm{equ2}}=[0.884.0.21]$. 
Among them, $x_{\mathrm{equ1}}$ and $x_{\mathrm{equ2}}$ are stable equilibrium points, while $x_{\mathrm{equ0}}$ is a saddle point.
Therefore, all state trajectories eventually reach either $x_{\mathrm{equ1}}$ or $x_{\mathrm{equ2}}$.

Fig. \ref{fig: tunnel expansion left} and Fig. \ref{fig: tunnel expansion right} demonstrate the expansion of $\epsilon$-controllable subset when target states are set to $x_{\mathrm{equ1}}$ and $x_{\mathrm{equ2}}$, respectively.
When the expansion originates from either of them, the $\epsilon$-controllable subsets do not intersect with those expanded from the other one.
The effect of $\epsilon$ on DOC is further discussed.
As depicted in Fig. \ref{fig: tunnel epsilon a}, DOC is independent of $\epsilon$ at three equilibrium points.
Besides, the sum of DOC at $x_{\mathrm{equ1}}$ and $x_{\mathrm{equ2}}$ approximately equals $1$, indicating that almost no state is $\epsilon$-controllable when the target state is the unstable equilibrium $x_{\mathrm{equ0}}$.
In this system, all states rapidly converge toward their nearest stable equilibrium point.

\begin{figure}
\begin{center}
    \subfloat[Step=$1$]{
        \includegraphics[trim=0 5 30 30, clip, width=0.48\linewidth]{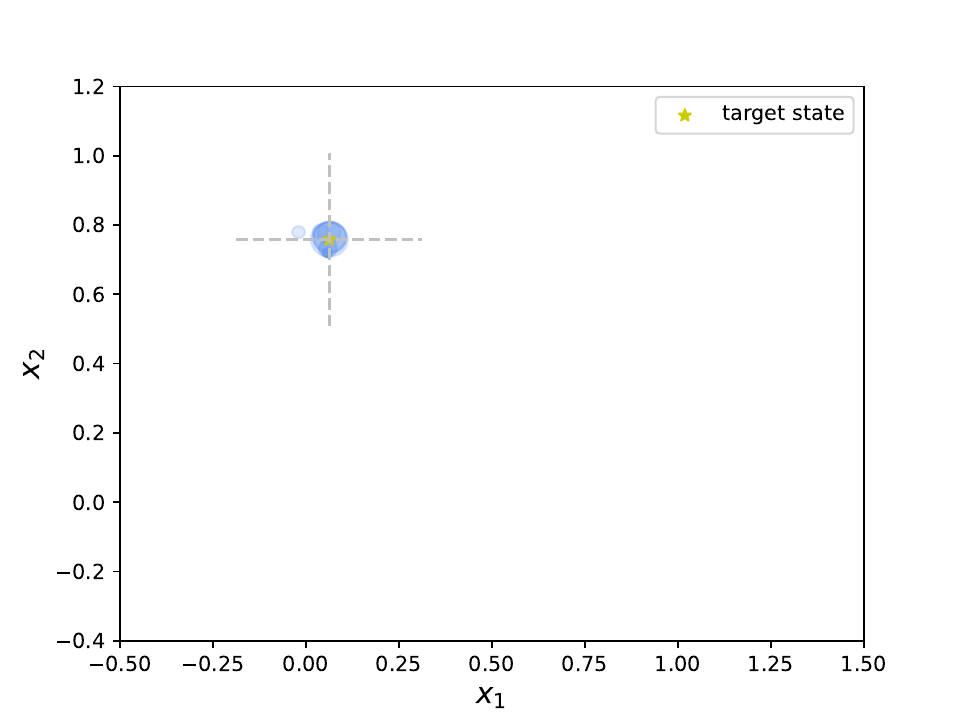}
    }
    \subfloat[Step=$200$]{
        \includegraphics[trim=0 5 30 30, clip, width=0.48\linewidth]{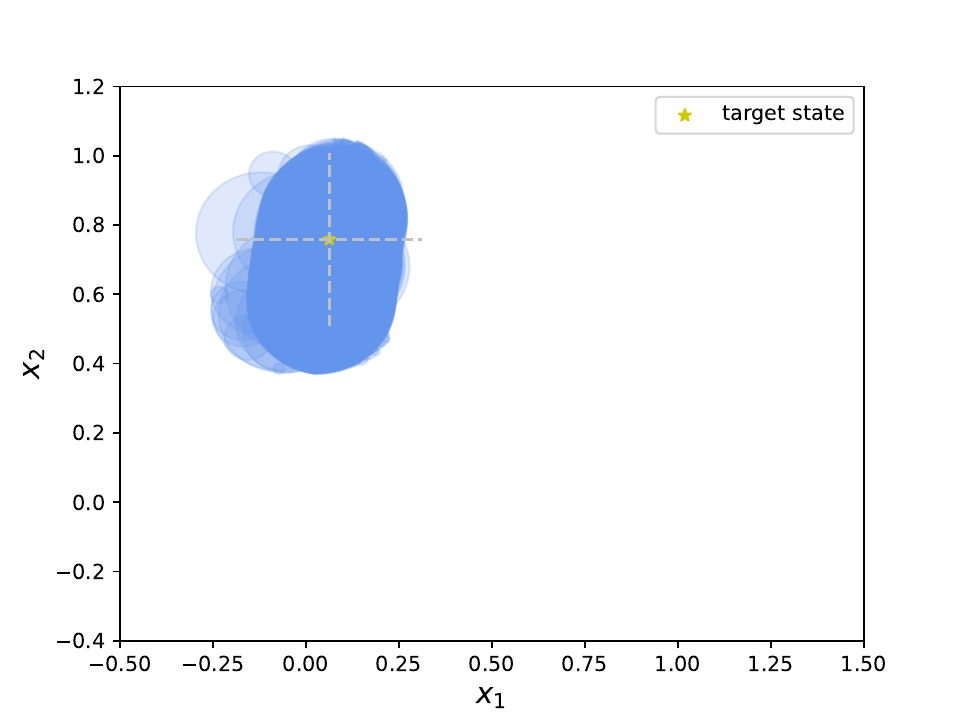}
    }\\
    \subfloat[Step=$400$]{
        \includegraphics[trim=0 5 30 30, clip, width=0.48\linewidth]{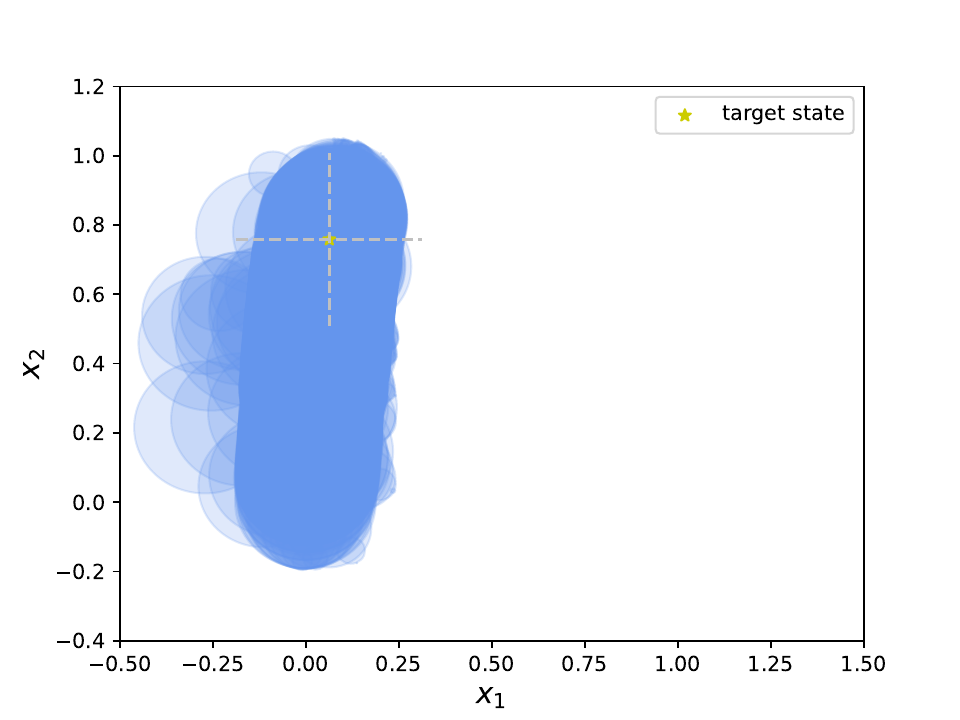}
    }
    \subfloat[Step=$600$]{
        \includegraphics[trim=0 5 30 30, clip, width=0.48\linewidth]{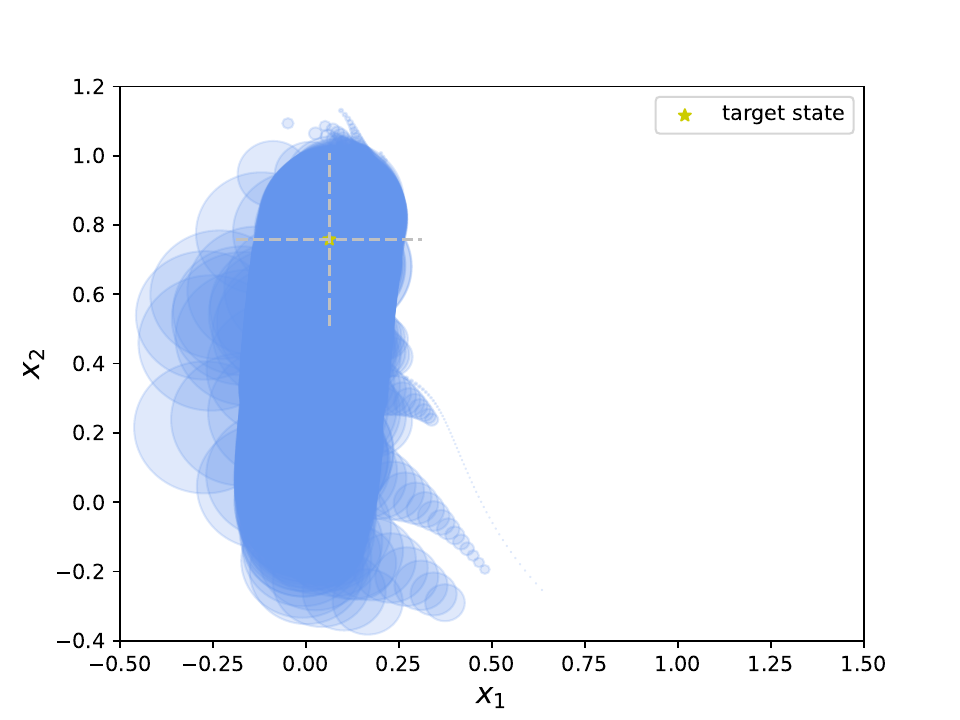}
    }\\
    \caption{The process of $\epsilon$-controllable set expansion in a tunnel-diode circuit when $x_\text{equ1}=[0.063,0.758]$.}
    \label{fig: tunnel expansion left}
\end{center}
\end{figure}

\begin{figure}
\begin{center}
    \subfloat[Step=$1$]{
        \includegraphics[width=0.48\linewidth]{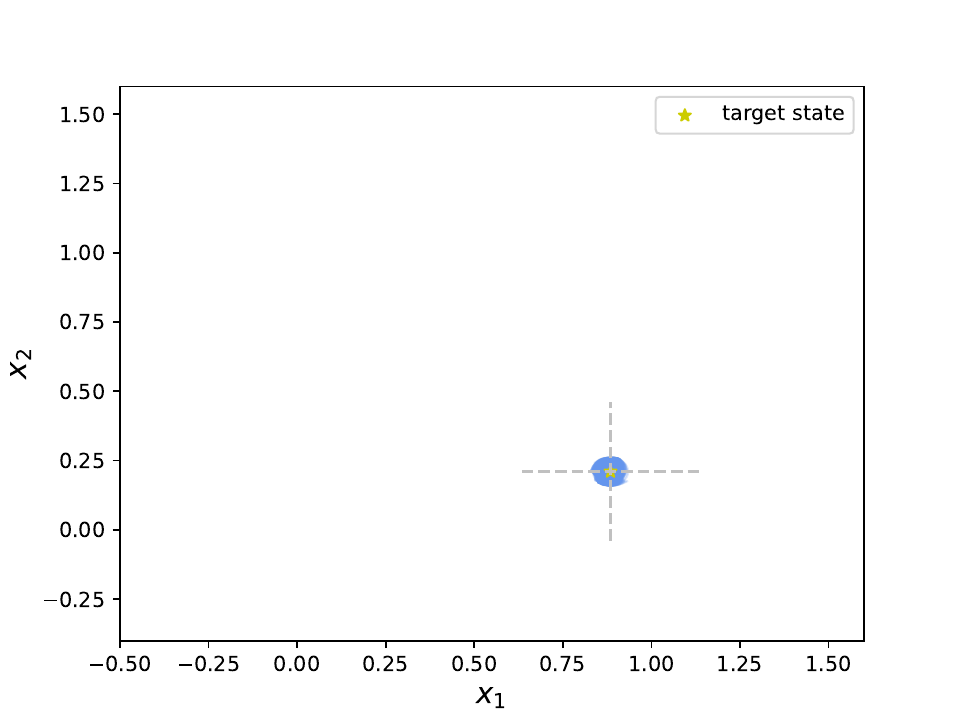}
    }
    \subfloat[Step=$1000$]{
        \includegraphics[width=0.48\linewidth]{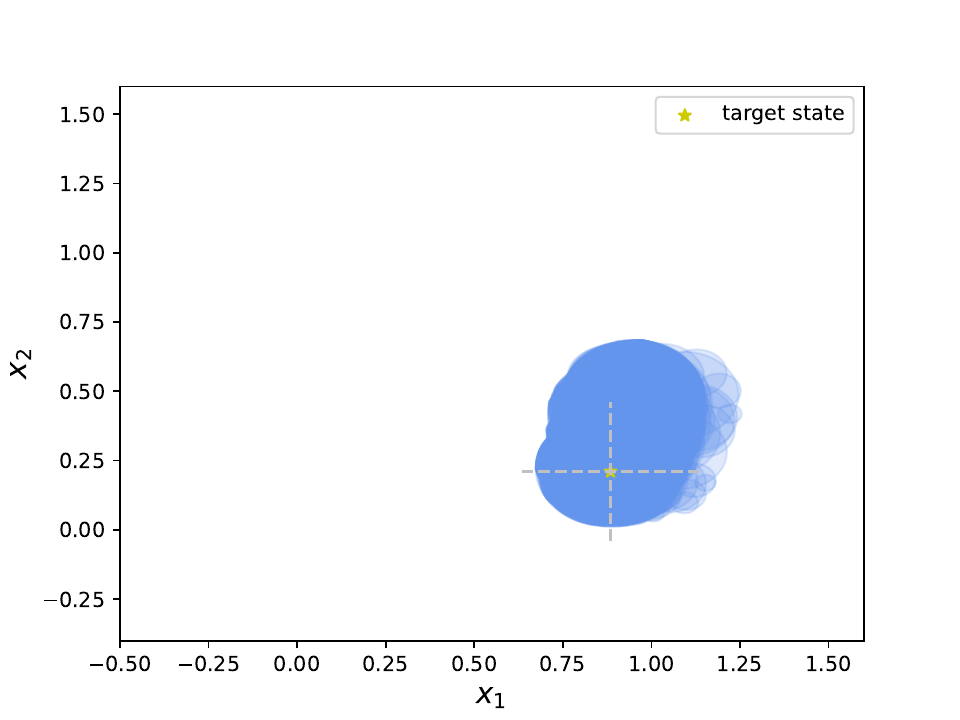}
    }\\
    \subfloat[Step=$2000$]{
        \includegraphics[width=0.48\linewidth]{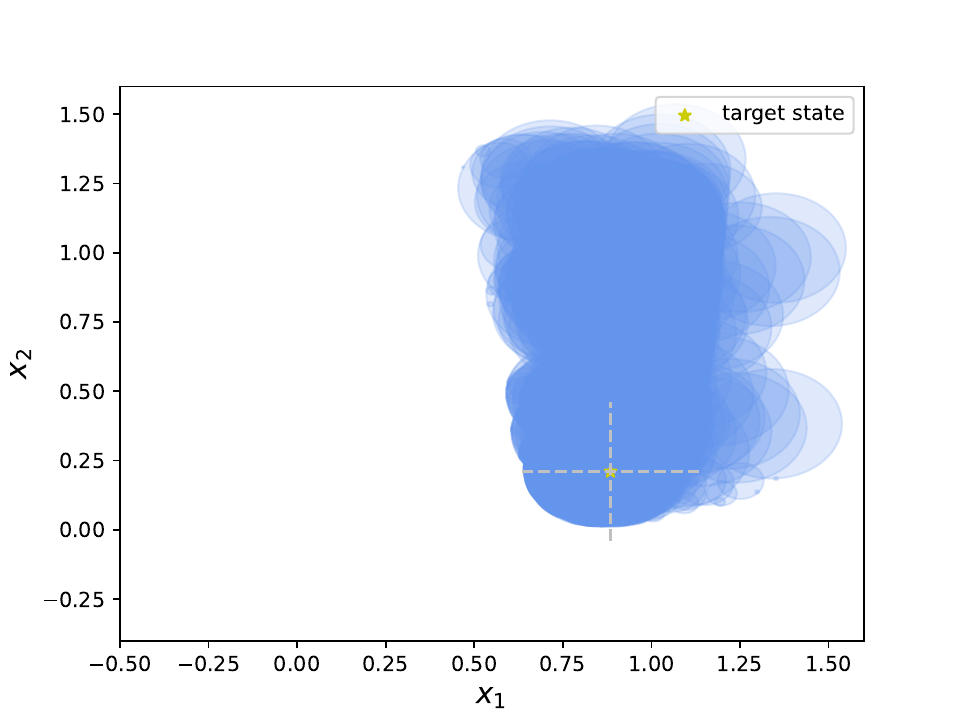}
    }
    \subfloat[Step=$3000$]{
        \includegraphics[width=0.48\linewidth]{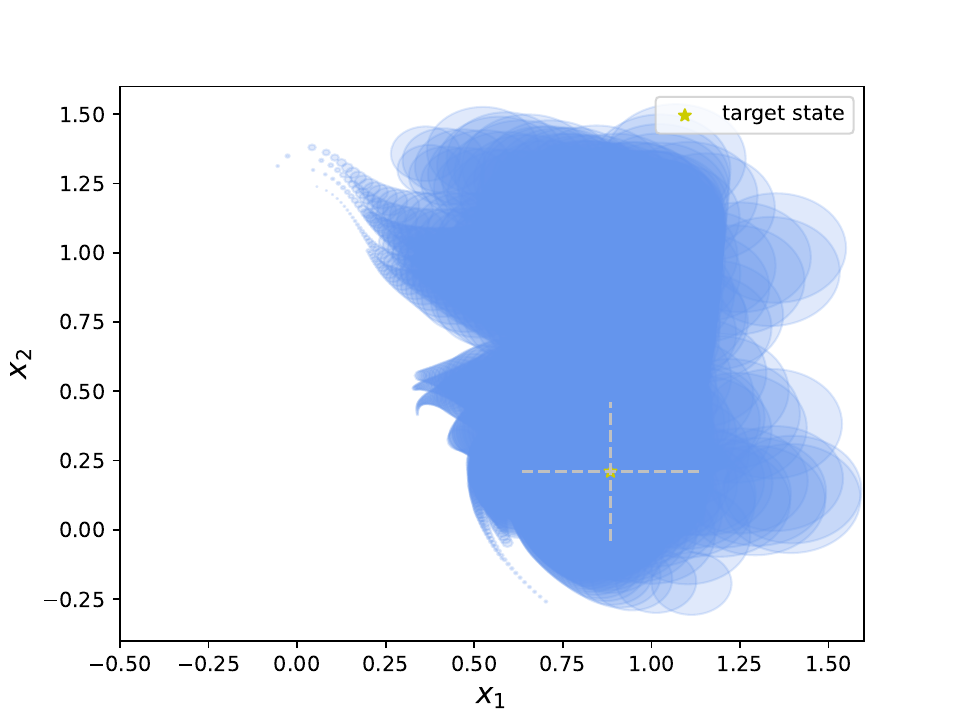}
    }\\
    \caption{The process of $\epsilon$-controllable set expansion in a tunnel-diode circuit when $x_{\mathrm{equ2}}=[0.884,0.21]$.}
    \label{fig: tunnel expansion right}
\end{center}
\end{figure}
\begin{figure}
  \begin{center}
  \subfloat[DOC changes with $\epsilon$]{\includegraphics[trim=10 0 15 10, clip, width=0.48\linewidth]{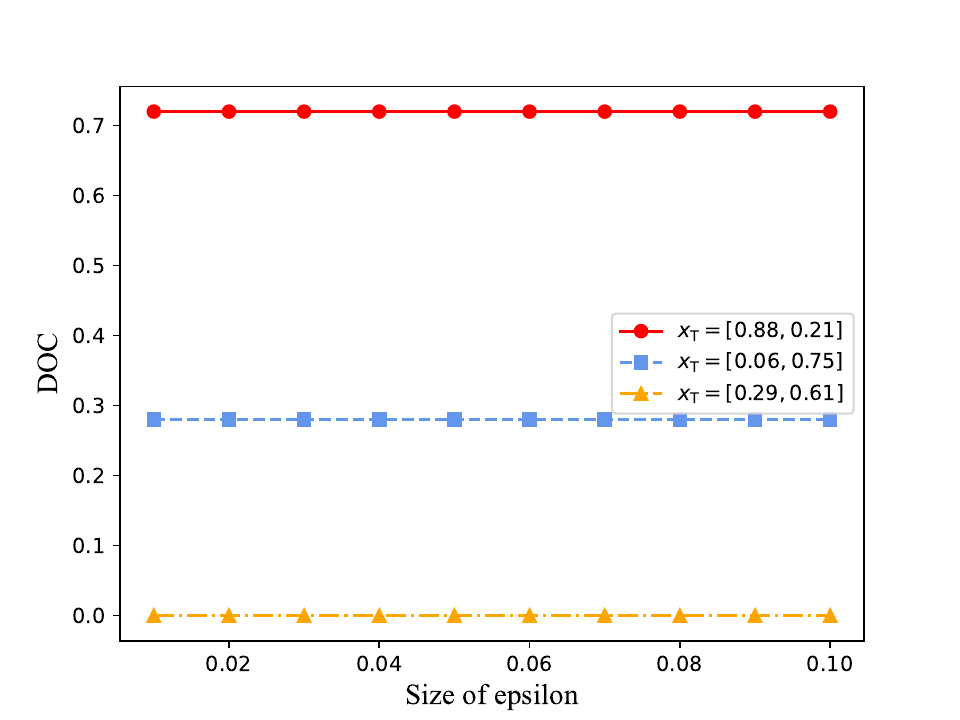}\label{fig: tunnel epsilon a}}
  \subfloat[DOC for different target states]{\includegraphics[trim=20 30 35 30, clip, width=0.48\linewidth]{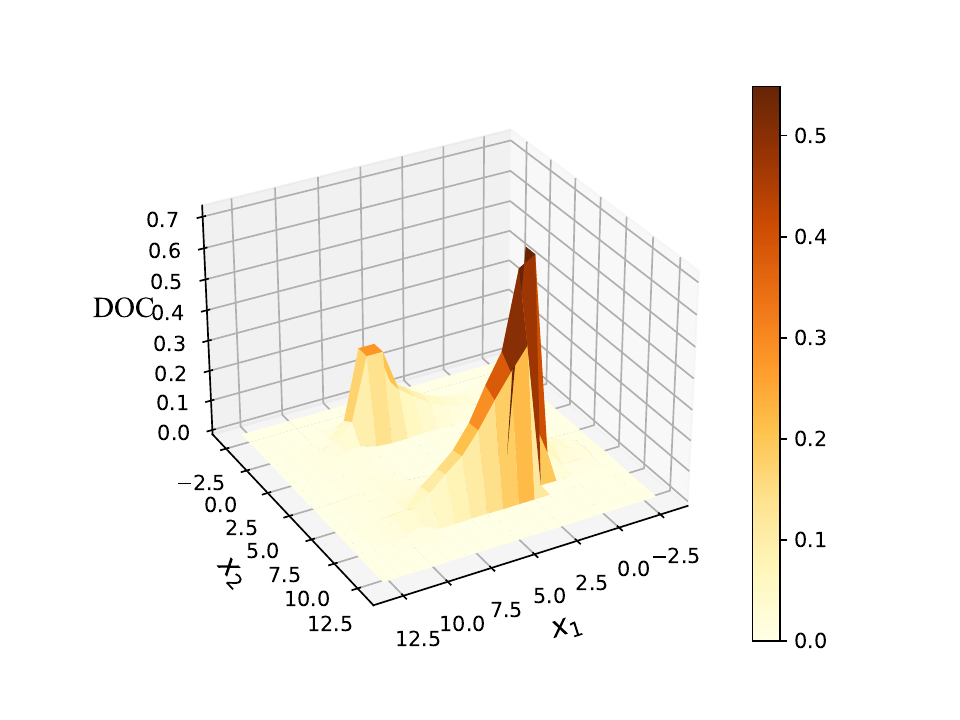}\label{fig: tunnel epsilon b}}
  \caption{The effect of different error radii and target states on DOC in a tunnel-diode circuit.}
  \label{fig: tunnel epsilon}
  \end{center}
\end{figure}

The results for DOC under different target states are visualized in Fig. \ref{fig: tunnel epsilon b}.
It shows that in this system, not all states are $\epsilon$-controllable with respect to a specific equilibrium point.
Only states in a surrounding area of each equilibrium point can be steered to it.
Beyond these areas, DOC is almost zero for any equilibrium point.
The DOCs for the two stable equilibrium points are $0.3$ and $0.7$, which is consistent with the fact that all states can only be steered to their nearest stable equilibrium point.
The final $\epsilon$-controllable subsets for two stable equilibrium points are visualized in Fig. \ref{fig: tunnel_results}.

\begin{figure}[!htbp]
  \begin{center}
    \subfloat[$x_{\text{T}}=x_\text{{equ1}}$]{
        \includegraphics[trim=0 5 30 30, clip,width=0.48\linewidth]{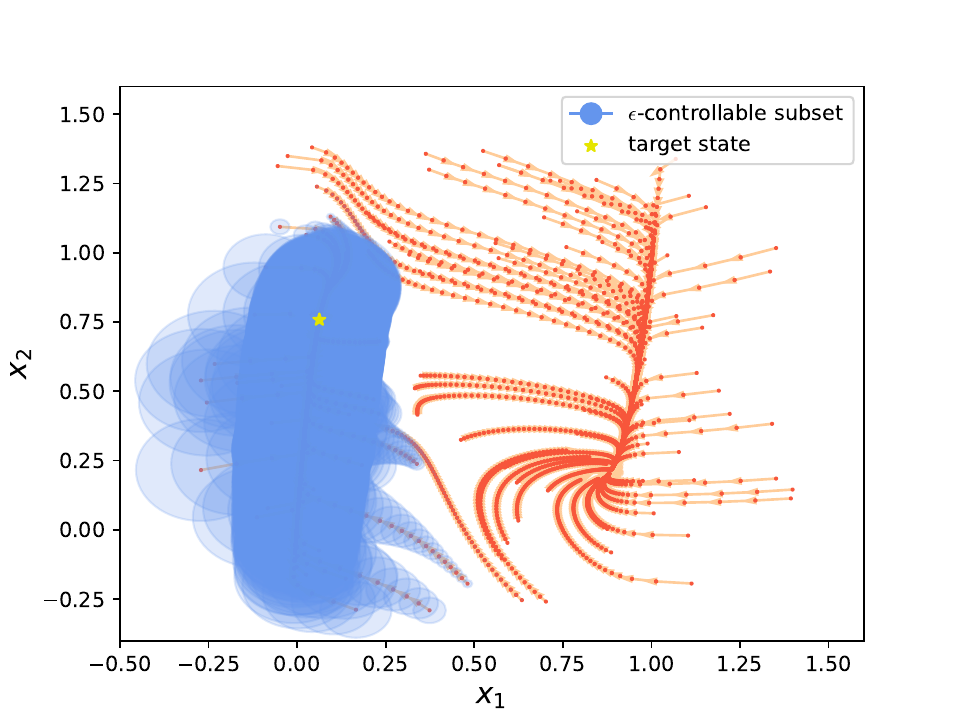}
    }
    \subfloat[$x_{\text{T}}=x_\text{{equ2}}$]{
        \includegraphics[trim=0 5 30 30, clip,width=0.48\linewidth]{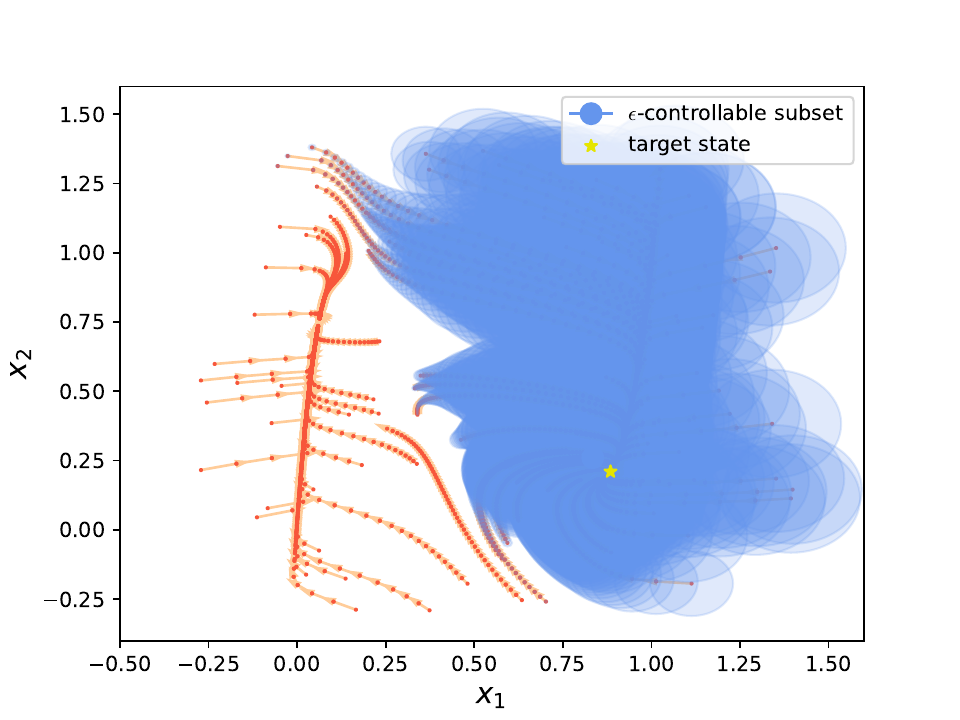}
    }\\
  \caption{The relationship between the final $\epsilon$-controllable subsets and sampled trajectories of a tunnel-diode circuit without control input when target states are $x_\text{equ1}$ and $x_\text{equ2}$.}
  \label{fig: tunnel_results}
  \end{center}
\end{figure}

In conclusion, MECS is applicable to both linear and nonlinear datatic systems.
When the target state coincides with the equilibrium point, the maximum number of $\epsilon$-controllable states is achieved.
The error radius represents a relaxation of traditional exact controllability, allowing states to reach only a neighborhood of the target state.
As the error radius increases, more states become $\epsilon$-controllable.
When the error radius tends to zero, $\epsilon$-controllability degenerates to traditional exact controllability.

\section{Conclusion}
This paper proposes a controllability test method for datatic (i.e., data-driven) control systems by introducing a new concept called $\epsilon$-controllability.
This concept extends the definition of controllability from traditional point-to-point form to a point-to-region form, making it more applicable for dynamic systems whose behaviors are described by a limited number of data points.
By leveraging the Lipschitz continuity assumption to extrapolate unknown state transition from data points, we establish the one-step controllability backpropagation theorem.
This theorem enables the expansion of controllability from a known $\epsilon$-controllable subset to a new one.
Based on this theorem, we propose the maximum expansion of controllable subset (MECS) algorithm to efficiently identify controllable states in nonlinear systems with datatic description.
MECS searches the $\epsilon$-controllable tree by iteratively performing four steps: selection, expansion, evaluation, and pruning until all $\epsilon$-controllable states are found.
To reduce the computational complexity of MECS, we propose a simplified algorithm called Floyd expansion with radius fixed (FERF) based on a mutual controllability assumption of neighboring states.
FERF maintains a fixed radius of all expanded balls and finds $\epsilon$-controllability states by solving a shortest path problem.
In many real-world systems, such as those with image state spaces, it is impractical to verify exact controllability of every individual state element.
Our two algorithms offer an alternative solution to numerically test their controllability, which requires defining specialized distance metrics between states that selectively focus on critical state elements.
The new definition of $\epsilon$-controllability and its testing methods establish a basic theory and provide a practical tool for controllability analysis of datatic control systems.

\ifCLASSOPTIONcaptionsoff
  \newpage
\fi



%\begin{thebibliography}{1}


\bibliographystyle{IEEEtran}
\bibliography{IEEEabrv,Bibliography}

\begin{thebibliography}{10}
\providecommand{\url}[1]{#1}
\csname url@rmstyle\endcsname
\providecommand{\newblock}{\relax}
\providecommand{\bibinfo}[2]{#2}
\providecommand\BIBentrySTDinterwordspacing{\spaceskip=0pt\relax}
\providecommand\BIBentryALTinterwordstretchfactor{4}
\providecommand\BIBentryALTinterwordspacing{\spaceskip=\fontdimen2\font plus
\BIBentryALTinterwordstretchfactor\fontdimen3\font minus \fontdimen4\font\relax}
\providecommand\BIBforeignlanguage[2]{{%
\expandafter\ifx\csname l@#1\endcsname\relax
\typeout{** WARNING: IEEEtran.bst: No hyphenation pattern has been}%
\typeout{** loaded for the language `#1'. Using the pattern for}%
\typeout{** the default language instead.}%
\else
\language=\csname l@#1\endcsname
\fi
#2}}
\renewcommand\BIBentryALTinterwordstretchfactor{4}

\bibitem{kalman1960on}
R.~Kalman, ``On the general theory of control systems,'' \emph{IFAC Proceedings Volumes}, vol.~1, no.~1, pp. 491--502, 1960, 1st International IFAC Congress on Automatic and Remote Control, Moscow, USSR, 1960.

\bibitem{trentelman2012control}
H.~L. Trentelman, A.~A. Stoorvogel, and M.~Hautus, \emph{Control theory for linear systems}.\hskip 1em plus 0.5em minus 0.4em\relax Springer Science \& Business Media, 2012.

\bibitem{kalman1963mathematical}
R.~E. Kalman, ``Mathematical description of linear dynamical systems,'' \emph{Journal of the Society for Industrial and Applied Mathematics, Series A: Control}, vol.~1, no.~2, pp. 152--192, 1963.

\bibitem{popov1966hiperstabilitatea}
V.~M. Popov, \emph{Hiperstabilitatea sistemelor automate}.\hskip 1em plus 0.5em minus 0.4em\relax Editura Academiei Republicii Socialiste Rom{\^a}nia, 1966.

\bibitem{gershwin1971controllability}
S.~Gershwin and D.~Jacobson, ``A controllability theory for nonlinear systems,'' \emph{IEEE Transactions on Automatic Control}, vol.~16, no.~1, pp. 37--46, Feb. 1971.

\bibitem{hermann1977nonlinear}
R.~Hermann and A.~Krener, ``Nonlinear controllability and observability,'' \emph{IEEE Transactions on Automatic Control}, vol.~22, no.~5, pp. 728--740, Oct. 1977.

\bibitem{yamamoto1977controllability}
Y.~Yamamoto, ``Controllability of nonlinear systems,'' \emph{Journal of Optimization Theory and Applications}, vol.~22, no.~1, pp. 41--49, May 1977.

\bibitem{yang2024stability}
Y.~Yang, Z.~Zheng, and S.~E. Li, ``On the stability of datatic control systems,'' \emph{arXiv preprint arXiv:2401.16793}, 2024.

\bibitem{zhan2024canonical}
G.~Zhan, Z.~Zheng, and S.~E. Li, ``Canonical form of datatic description in control systems,'' \emph{arXiv preprint arXiv:2403.01768}, 2024.

\bibitem{hussein2017imitation}
A.~Hussein, M.~M. Gaber, E.~Elyan, and C.~Jayne, ``Imitation learning: A survey of learning methods,'' \emph{ACM Computing Surveys (CSUR)}, vol.~50, no.~2, pp. 1--35, 2017.

\bibitem{ho2016generative}
J.~Ho and S.~Ermon, ``Generative adversarial imitation learning,'' \emph{Advances in neural information processing systems}, vol.~29, 2016.

\bibitem{li2023reinforcement}
S.~E. Li, \emph{Reinforcement learning for sequential decision and optimal control}.\hskip 1em plus 0.5em minus 0.4em\relax Springer Verlag, Singapore, 2023.

\bibitem{guan2021direct}
Y.~Guan, S.~E. Li, J.~Duan, J.~Li, Y.~Ren, Q.~Sun, and B.~Cheng, ``Direct and indirect reinforcement learning,'' \emph{International Journal of Intelligent Systems}, vol.~36, no.~8, pp. 4439--4467, 2021.

\bibitem{zhuowang2011databased}
Z.~Wang and D.~Liu, ``Data-{{Based Controllability}} and {{Observability Analysis}} of {{Linear Discrete-Time Systems}},'' \emph{IEEE Transactions on Neural Networks}, vol.~22, no.~12, pp. 2388--2392, Dec. 2011.

\bibitem{liu2014databased}
D.~Liu, P.~Yan, and Q.~Wei, ``Data-based analysis of discrete-time linear systems in noisy environment: Controllability and observability,'' \emph{Information Sciences}, vol. 288, pp. 314--329, 2014.

\bibitem{shaker2017new}
H.~R. Shaker and S.~{Lazarova-Molnar}, ``A new data-driven controllability measure with application in intelligent buildings,'' \emph{Energy and Buildings}, vol. 138, pp. 526--529, Mar. 2017.

\bibitem{vanwaarde2020data}
H.~J. Van~Waarde, J.~Eising, H.~L. Trentelman, and M.~K. Camlibel, ``Data {{Informativity}}: {{A New Perspective}} on {{Data-Driven Analysis}} and {{Control}},'' \emph{IEEE Transactions on Automatic Control}, vol.~65, no.~11, pp. 4753--4768, Nov. 2020.

\bibitem{mishra2021datadriven}
V.~K. Mishra, I.~Markovsky, and B.~Grossmann, ``Data-{{Driven Tests}} for {{Controllability}},'' \emph{IEEE Control Systems Letters}, vol.~5, no.~2, pp. 517--522, Apr. 2021.

\bibitem{cormen2022introduction}
T.~H. Cormen, C.~E. Leiserson, R.~L. Rivest, and C.~Stein, \emph{Introduction to algorithms}.\hskip 1em plus 0.5em minus 0.4em\relax MIT press, 2022.

\bibitem{bentley1975multidimensional}
J.~L. Bentley, ``Multidimensional binary search trees used for associative searching,'' \emph{Communications of the ACM}, vol.~18, no.~9, pp. 509--517, 1975.

\bibitem{ye1989extension}
Y.~Ye and E.~Tse, ``An extension of karmarkar's projective algorithm for convex quadratic programming,'' \emph{Mathematical programming}, vol.~44, pp. 157--179, 1989.

\end{thebibliography}

\vfill


\end{document}